\documentclass[a4paper,11pt]{article}

\usepackage{graphicx}%
\usepackage{subcaption}
\usepackage{authblk}
\usepackage{amsmath}
\usepackage{amssymb}
\usepackage{amsfonts}

\providecommand{\keywords}[1]{\textit{Keywords:} #1}

\newtheorem{theorem}{Theorem}

\newtheorem{claim}[theorem]{Claim}

\newtheorem{lemma}[theorem]{Lemma}

\newtheorem{remark}[theorem]{Remark}

\newenvironment{proof}[1][Proof]{\noindent\textbf{#1.} }{\ \rule{0.5em}{0.5em}}

\begin{document}

\title{Diffusion properties of small-scale fractional transport models}                      

\author[1]{Paolo Cifani}
\author[1]{Franco Flandol}
\affil[1]{Scuola Normale Superiore, Piazza dei Cavalieri, 7, Pisa, Italy}

\date{\today}

\maketitle

\begin{abstract}
Stochastic transport due to a velocity field modeled by the superposition of
small-scale divergence free vector fields activated by Fractional Gaussian
Noises (FGN) is numerically investigated. We present two non-trivial
contributions: the first one is the definition of a model where different
space-time structures can be compared on the same ground: this is achieved by
imposing the same average kinetic energy to a standard Ornstein-Uhlenbeck
approximation, then taking the limit to the idealized white noise structure.
The second contribution, based on the previous one, is the discover that a
mixing spatial structure with persistent FGN in the Fourier components induces
a classical Brownian diffusion of passive particles, with suitable diffusion
coefficient; namely, the memory of FGN is lost in the space complexity of the
velocity field.
\end{abstract}


\maketitle

\section{Introduction}\label{sec:into}

Many phenomena in nature, most prominently the motion of particles suspended in a quiescent medium, are well described by standard Brownian Motion $B_t$. It is often found that the physics of these problems is statistically stationary. The prototypical model for such phenomena is the Ornstein-Uhlenbeck (OU) process
\begin{equation}
dZ_{t}^{\tau}=-\frac{1}{\tau}Z_{t}^{\tau} dt+\left(\frac{2\sigma^2}{\tau}\right)^{1/2}dB_{t}
\label{eq:OU_basic}
\end{equation}
with integral timescale $\tau$ and stationary variance $\sigma^2$. The increments $dB_t$ are Gaussian white noise, thus samples of the latter at different times are independent. While this is a good approximation in several instances, studies on random processes, e.g. turbulent flows and financial time series, have shown strong interdependence between distant samples. To this aim, an extension to (\ref{eq:OU_basic}) was put forward in the seminal paper \cite{mandelbrot} where long-range dependence is regulated by the Hurst exponent $H \in ]0,1[$. We can speak then of fractional Ornstein-Uhlenbeck process
\begin{equation}
dZ_{t}^{\tau,H}=-\frac{1}{\tau}Z_{t}^{\tau,H} dt+\frac{c_H}{\tau^H}dB_{t}^H
\label{eq:OU_fractional}
\end{equation}
where the driving random process $B_t^H$ is a fractional Gaussian Process of Hurst exponent $H$. Despite is early appearance about half a century ago, the literature on this subject is relatively young, primarily due to the difficulties, both analytically and numerically, introduced by statistical dependence of increments. Here, by means of theoretical and numerical tools, we attempt to take a step forward into the understanding of the fate of particles transported by vector fields whose components are Fractional Gaussian Processes. In particular, we will focus on stochastic transport for $H > 1/2$, i.e. positively correlated increments, and compare the findings to the case $H=1/2$, i.e. standard Brownian Motion. 

For a rigorous definition of Fractional Brownian Motion and stochastic integration in
the case $H\geq1/2$ (Young's integral) and its link with Stratonovich
integration, we refer the reader to \cite{Nualart} (also \cite{Kruk}, or \cite{FlaRusso}).
Fractional Gaussian Processes in applications, such as turbulent fluid models,
have been introduced in several works. In most cases the fractality is however
understood with respect to the space-structure (see for instance
\cite{Apollinaire}, \cite{Aplooin2}), because of its great interest in
connection with Kolmogor theory and variants like the multifractal model. The
interest of Fractional Brownian Motion in time for turbulence modeling is
discussed for instance in \cite{Chevillard}, see also references therein.
Several very interesting works prove a fractional structure of the limit
process of an homogenization procedure, rescaling of deterministic or
stochastic fields; the space structure is never "chaotic" in the sense of the
present paper, hence the emergence of a fractional behaviour in the limit; see
for instance \cite{FannKom0}, \cite{FannKom}, \cite{FannKom2}. 

The phenomenon considered here seems to be new: we found the emergence of a
Brownian behaviour in time from a space-time structure consisting of Fractional Brownian Motions in time and spatial high 
frequency fluctuations in all directions, which restore some independence of
increments. The model considered here is similar to the one theoretically
investigated in \cite{FlaRusso}, where closed forms of moments of solution are
found. However the case of non commuting vector fields - precisely the case
which restores a degree of independence of increments - has not been
theoretically solved there, only preliminary discussed, and indeed the result
of the present paper is a confirmation of the fact that the time behavior is
not trivial, in the non-commutative case. The model of \cite{FlaRusso} has
some similarity with the model considered in \cite{KomorNovikov}, where
however only two Fractional Brownian Motions act, hence the restoring of
independence is not possible. See also \cite{Squarcini} for a model with some
similar features.

The model considered here and in \cite{FlaRusso} is an extension to Fractional
Brownian Motions of the models considered in several works in the case of
classical Brownian Motion, see for instance \cite{Galeati}, \cite{FlaLuongo},
\cite{Luo} among several works also cited there. These papers deal with
stochastic transport in Stratonovich form, a basic modeling idea performed
recently for several models, also for small-scale transport of large scales -
not only for transport of a passive scalars - see for instance \cite{Chapron},
\cite{Crisan1}, \cite{Crisan2}, \cite{Ephrati}, \cite{Holm}, \cite{Holm2},
\cite{Memin}, \cite{Resseguier}; see also \cite{Majda} for a review of
diffusion limits.

The structure of this paper is organised as follows: in Sec. \ref{sec:stoc_struct} the analytical framework is presented and the derivation of the our stochastic model is given. Specific examples of the model are then illustrated and a statement of the main claim of this work is provided. In Sec. \ref{sec:numerical} the numerical results are presented and compared with the theoretical predictions. Finally, in Sec. \ref{sec:concl} conclusions and outlook are summarised. The analytical derivations are collected in the Appendix to facilitate the readability of this paper. 

\section{Stochastic transport structure}\label{sec:stoc_struct}

Consider the transport equation%
\begin{equation}
\begin{aligned}
\partial_{t}T+\mathbf{u}\cdot\nabla T  &  =0\\
T|_{t=0}  &  =T_{0}%
\end{aligned}
\label{eq:system}
\end{equation}
in $\mathbb{R}^{2}$, where $\mathbf{u}\left(  \mathbf{x},t\right)  $ is a
divergence free vector field. Assume that $T_{0}\geq0$ is integrable, or more
conventionally that it is a probability density function (pdf), so that
$T\left(  \cdot,t\right)  $ is also a pdf for every $t\geq0$ (since
$\mathbf{u}$ will have the necessary regularity for such a result).

We assume that $\mathbf{u}\left(  \mathbf{x},t\right)  $ is not the true
solution of a fluid dynamic equation but it is a stochastic model preserving
some idealized properties of a turbulent fluid, precisely a model of the
following simplified form%
\begin{equation}
\mathbf{u}\left(  \mathbf{x},t\right)  =uC\left(  \eta,\tau,H\right)
\sum_{\mathbf{k}\in\mathbf{K}_{\eta}}\mathbf{\sigma}_{\mathbf{k}}\left(
\mathbf{x}\right)  \frac{dB_{t}^{H,\mathbf{k}}}{dt}%
\label{eq:model}
\end{equation}
where $C\left(  \eta,\tau,H\right)  $ is a normalizing constant allowing us to
compare models with different space-time structure. Here $u$ is an
\textit{average velocity} constant, with dimension $\left[  L\right]  /\left[
T\right]  $; $\eta$ is a space scale (inspired to the notation of the
so-called Kolmogorov scale), with dimension $\left[  L\right]  $; $H$ is the
Hurst index of the independent real-valued Fractional Brownian Motions (FBM)
$B_{t}^{H,\mathbf{k}}$, which have dimension $\left[  T\right]  ^{H}$ (due to
the property $\mathbb{E}\left[  \left\vert B_{t}^{H,\mathbf{k}}\right\vert
^{2}\right]  =t^{2H}$); the index set $\mathbf{K}_{\eta}$ will correspond (in
the nontrivial case) to \textit{length scales of order} $\eta$, and it is
assumed to be a finite set; the divergence free vector fields $\mathbf{\sigma
}_{\mathbf{k}}\left(  \mathbf{x}\right)  $ will be described below in the
examples, and are dimensionless. The normalizing constant $C\left(
\eta,H\right)  $ has dimension $\left[  T\right]  ^{1-H}$, to compensates the
dimension $\left[  T\right]  ^{H-1}$ of $\frac{dB_{t}^{H,\mathbf{k}}}{dt}$.
Precisely, the constant $C\left(  \eta,\tau,H\right)  $ is given by%
\[
C\left(  \eta,\tau,H\right)  =\frac{\tau^{1-H}\sqrt{2}}{\sqrt{\Gamma\left(
2H+1\right)  }}\frac{1}{C_{\eta}}%
\]
where $\tau$ has the meaning of \textit{relaxation time} of the fluid, and it
is typically a small constant, with dimension $\left[  T\right]  $,
$\Gamma\left(  r\right)  $ is the Gamma function and the constant $\frac
{1}{C_{\eta}}$ is a normalizing factor for the sum over $\mathbf{k}$ of the
$\mathbf{\sigma}_{\mathbf{k}}$, defined by (\ref{def of C eta}) below. We give a precise
motivation for the choice of the noise and all the constants in its definition in Appendix \ref{appendix noise structure}.

Thanks to the factor $u\tau^{1-H}$, we keep memory of the fact that a true
fluid has a finite relaxation time and a finite kinetic energy, properties
that are formally lost in the model above. Indeed, the Fractional Gaussian
Noise (FGN) $\frac{dB_{t}^{H,\mathbf{k}}}{dt}$ does not have a characteristic
time-scale and has infinite variance. 

Thanks to the precise normalizing factor $C\left(  \eta,\tau,H\right)  $, we
may compare quantitatively different values of $H$ and $\eta$ (see Appendix
\ref{appendix noise structure}). The natural idea to put different models of
the previous form on the same ground would be to impose that they have the
same average kinetic energy; but, as we have already remarked, the FGN
$\frac{dB_{t}^{H,\mathbf{k}}}{dt}$ has infinite variance. Hence we introduce
an Ornstein-Uhlenbeck approximation%
\begin{equation}
\mathbf{u}_{\tau}\left(  \mathbf{x},t\right)  =\frac{u}{C_{\eta}}%
\sum_{\mathbf{k}\in\mathbf{K}_{\eta}}\mathbf{\sigma}_{\mathbf{k}}\left(
\mathbf{x}\right)  Z_{t}^{\tau,H,\mathbf{k}}\label{smooth velocity}%
\end{equation}
where $Z_{t}^{\tau,H,\mathbf{k}}$ is the solution of equation
\begin{equation}
dZ_{t}^{\tau,H,\mathbf{k}}=-\frac{1}{\tau}Z_{t}^{\tau,H,\mathbf{k}}%
dt+\frac{c_{H}}{\tau^{H}}dB_{t}^{H,\mathbf{k}}\label{OU}%
\end{equation}
with $Z_{0}^{\tau,H,\mathbf{k}}=0$ and $c_{H}$ chosen so that $\mathbb{E}\left[  \left\vert Z_{t}^{\tau
,H,\mathbf{k}}\right\vert ^{2}\right]  \to 1$ as $t\to \infty$ (the factor $\tau^{H}$ in the noise
term $\frac{c_{H}}{\tau^{H}}dB_{t}^{H,\mathbf{k}}$ compensate the dimension of
$B_{t}^{H,\mathbf{k}}$ to produce the adimensional quantity $Z_{t}%
^{\tau,H,\mathbf{k}}$); it is given by $c_{H}=\frac{\sqrt{2}}{\sqrt
{\Gamma\left(  2H+1\right)  }}$ (see Appendix \ref{appendix noise structure}). As shown in Lemma \ref{Lemma_main}, the previous
equation can be approximated by
\begin{equation}
Z_{t}^{\tau,H,\mathbf{k}}\sim\tau^{1-H}c_{H}dB_{t}^{H,\mathbf{k}}%
\label{eq:OU_surr}
\end{equation}
which leads to the model above. Process (\ref{eq:OU_surr}) is more amenable to 
analytical treatment than (\ref{OU}) and therefore adopted in this work as a surrogate of the OU process. 
Moreover, we consider this "normalization" an important
step in view of the comparison between different forms of stochastic
transport. 

Concerning the smooth divergence free vector fields $\mathbf{\sigma
}_{\mathbf{k}}$, we assume that the limits
\begin{equation}
\left\langle \mathbf{\sigma}_{\mathbf{k}}\right\rangle ^{2}:=\lim
_{R\rightarrow\infty}\frac{1}{R^{2}}\int_{\left[  -\frac{R}{2},\frac{R}%
{2}\right]  ^{2}}\left\vert \mathbf{\sigma}_{\mathbf{k}}\left(  \mathbf{x}%
\right)  \right\vert ^{2}dx\label{condition on sigma}%
\end{equation}
exist, and take $C_{\eta}$ above given by%
\begin{equation}
C_{\eta}=\sqrt{\sum_{\mathbf{k}\in\mathbf{K}_{\eta}}\left\langle
\mathbf{\sigma}_{\mathbf{k}}\right\rangle ^{2}}.\label{def of C eta}%
\end{equation}

\subsection{Specific examples and motivation}

In the choice of the two examples below we are motivated by a certain
variety of turbulent flows appearing in confined plasma experiments and
simulations. We will not treat realistic velocity fields emerging from such
application but only paradigmatic idealizations, however inspired by such
observations. 

It is observed that, in the poloidal section, the electromagnetic field, which
originally is perturbed at a very small scale by certain instabilities,
becomes organized also in structures of vortical type having a larger coherent scale. 
A stochastic parametrization of them could involve FGN with
Hurst parameter $H>1/2$, to model the persistence of the perturbation (the
larger a structure is, the more persistent is its transport effect). See for
instance, in the review \cite{garbet2010gyrokinetic}, Figures 9, 11, 14:
sometimes the perturbations are very disordered, sometimes else they are
organized in relatively parallel stripes (streamers, transport barriers). What
happens to heat and matter subject to such a velocity field? Which are the
turbulent transport properties?

We then consider two paradigmatic cases. The first one is simply made of
constant vector fields; it is not very realistic w.r.t. such applications but
it serves as a reference case. Moreover, even if so abstract, it behaves similarly to ``streamers", 
namely coherent elongated structures. The second one is made of several disordered
small-scale structures. Let us introduce the formal definitions.

The trivial case, that we call \textit{control case}, discussed mostly for
comparison, is defined by
\[
\mathbf{K}_{\eta}=\left\{  1,2\right\}  ,\qquad\mathbf{\sigma}_{1}\left(
\mathbf{x}\right)  =\left(  1,0\right)  ,\mathbf{\sigma}_{2}\left(
\mathbf{x}\right)  =\left(  0,1\right)  .
\]
We have $\left\langle \mathbf{\sigma}_{i}\right\rangle ^{2}=1$ for both
$i=1,2$, hence $C_{\eta}=\sqrt{2}$,%
\[
C\left(  \eta,\tau,H\right)  =\frac{\tau^{1-H}}{\sqrt{\Gamma\left(
2H+1\right)  }}%
\]
\[
\mathbf{u}\left(  \mathbf{x},t\right)  =u\frac{\tau^{1-H}}{\sqrt{\Gamma\left(
2H+1\right)  }}\left(  \frac{dB_{t}^{H,1}}{dt},\frac{dB_{t}^{H,2}}{dt}\right)
.
\]

We then introduce the so-called \textit{test case}, defined by a number
$\eta>0$,
\begin{equation}
\begin{aligned}
&\mathbf{K}_{\eta}=\left\{  \mathbf{k}\in\mathbb{Z}^{2}:\left\vert
\mathbf{k}\right\vert \in\left[  \frac{1}{2\eta},\frac{1}{\eta}\right]
\right\} \\
&\mathbf{\sigma}_{\mathbf{k}}\left(  \mathbf{x}\right)  =\frac{\mathbf{k}%
^{\perp}}{\left\vert \mathbf{k}\right\vert }\cos\left(  \mathbf{k}%
\cdot\mathbf{x}\right)  \qquad\text{if }\mathbf{k\in K}_{\eta}^{+}\text{ }%
\\
&\mathbf{\sigma}_{\mathbf{k}}\left(  \mathbf{x}\right)  =\frac{\mathbf{k}%
^{\perp}}{\left\vert \mathbf{k}\right\vert }\sin\left(  \mathbf{k}%
\cdot\mathbf{x}\right)  \qquad\text{if }\mathbf{k\in K}_{\eta}^{-}%
\end{aligned}
\label{eq:sigma_test_case}
\end{equation}

where $\mathbf{K}_{\eta}^{+}$ is the set of $\mathbf{k}=\left(  k_{1}%
,k_{2}\right)  \in\mathbf{K}_{\eta}$ such that either $\left\{  k_{1}%
>0\right\}  $ or $\left\{  k_{1}=0,k_{2}>0\right\}  $, and $\mathbf{K}_{\eta
}^{-}=-\mathbf{K}_{\eta}^{+}$. In this example, for each $\mathbf{k}$, one
has
\[
\left\langle \mathbf{\sigma}_{\mathbf{k}}\right\rangle ^{2}=\frac{1}{2\pi}%
\int_{0}^{2\pi}\sin^{2}tdt=\frac{1}{2}%
\]
hence
\begin{equation}
C_{\eta}=\sqrt{\frac{Card\left(  \mathbf{K}_{\eta}\right)  }{2}}\sim
\frac{\sqrt{3\pi}}{2\sqrt{2}\eta}\label{dubbio}%
\end{equation}
for small $\eta$ (because $Card\left(  \mathbf{K}_{\eta}\right)  \sim\pi
\frac{1}{\eta^{2}}-\pi\frac{1}{4\eta^{2}}=\pi\frac{3}{4\eta^{2}}$). 

As already said, we choose to describe the persistency of the structures by
independent FGN processes $\frac{dB_{t}^{H,\mathbf{k}}}{dt}$ with
\[
H\geq\frac{1}{2}.
\]
Concerning the interpretation of the product rule $\mathbf{u}\cdot\nabla T$
and the analogous product rule in equation (\ref{SDE}) below, if $H=\frac
{1}{2}$ (case of Brownian Motion) we use Stratonovich interpretation; if
$H>\frac{1}{2}$, we use Young integrals, which also are Stratonovich
integrals, in a sense, if compared to Skorohod ones (see [Nualart]). In both
cases, when necessary, we use the notation $\circ$ to recall that we use
Stratonovich interpretation.

\subsection{Transported quantities}

By solution $T\left(  \mathbf{x},t\right)  $ of the transport equation above
we mean, by definition, the stochastic process $T\left(  \mathbf{x},t\right)
$ uniquely identified by the formula%
\begin{equation}
T\left(  \mathbf{X}_{t}^{\mathbf{x}},t\right)  =T_{0}\left(  \mathbf{x}%
\right)  \label{Lagrangian formulation}%
\end{equation}
where $\mathbf{X}_{t}^{\mathbf{x}}$ is the solution of the equations of
characteristics%
\begin{equation}
d\mathbf{X}_{t}^{\mathbf{x}}=uC\left(  \eta,\tau,H\right)  \sum_{\mathbf{k}%
\in\mathbf{K}_{\eta}}\mathbf{\sigma}_{\mathbf{k}}\left(  \mathbf{X}%
_{t}^{\mathbf{x}}\right)  \circ dB_{t}^{H,\mathbf{k}},\qquad\mathbf{X}%
_{0}^{\mathbf{x}}=\mathbf{x.}\label{SDE}%
\end{equation}
The Lagrangian formulation can be proved to be equivalent to the SPDE above in
many cases. Here, for simplicity, we take it as the starting point, also
because we shall use numerical methods based on the Lagrangian formulation.

The key information we are interested in is how fast the information is
spread, diffused, by the velocity field. Therefore the key indicator is the
function%
\begin{equation}
t\mapsto\mathbb{E}\left[  \left\vert \mathbf{X}_{t}^{\mathbf{0}}\right\vert
^{2}\right]  .\label{key quantity}%
\end{equation}
In Appendix \ref{sec:app_b} we discuss the more general problem of understanding $T\left(
\mathbf{x},t\right)  $, but we restrict the numerical simulations and the
result to the quantity (\ref{key quantity}). Moreover, we give some
theoretical a priori information on some of these quantities, that can be used
to check the validity of the numerical simulations.

\subsection{Main results}\label{sec:main_res}

We here outline a synthetic description of the main results.
They will be described in detail in Section \ref{sec:numerical}.

As a preliminary step we simulate the process $\mathbf{X}_{t}^{\mathbf{0}}$, solution of
equation (\ref{SDE}), in the \textit{control case}, we obviously get a FBM, as
also described below in Appendix \ref{sec:app_b}. We use the control case to compute (\ref{key quantity}) for $H>1/2$ 
and validate our numerical simulation against theory. 

The main discovery of this paper is that, when we simulate the \textit{test
case}, with $H>1/2$, after a short transient the process behaves like a
Brownian Motion. The memory related to $H$ is lost. A trace of $H$ remains in
the diffusion coefficient:%
\[
\sigma^{2}\left(  \eta,\tau,H\right)  \sim\frac{\mathbb{E}\left[
\left\vert \mathbf{X}_{t}^{\mathbf{0}}\right\vert ^{2}\right]  }{t}\text{ for
}t\text{ large enough.}%
\]
We have not found a theoretical proof of this fact until now, but the heuristic reason is
relatively clear: the particle $\mathbf{X}_{t}^{\mathbf{0}}$ feels at time some components 
of the noise more than others, and changes the most
relevant components frequently, in its erratic motion. But different
components have independent processes: this restores a form of independence of
the increments, like for Brownian Motion. The numerical simulations of the present paper seem to indicate the validity
of the following theoretical result. 
\begin{theorem}\label{theo_main}
Given the value of all other coefficients, choose%
\[
\tau=\tau_{\eta}=C\eta^{\frac{1-2H}{1-H}}%
\]
for some constant $C>0$. Then the process $\mathbf{X}_{t}^{\mathbf{0}}$, which
depends on $\eta>0$, converges in law to a 2-dimensional Brownian Motion in
the limit as $\eta\rightarrow0$.
\end{theorem}
\begin{remark}
Notice that the factor $\tau_{\eta}^{\frac{1}{2H}-\frac{1}{2}}\eta^{1-\frac
{1}{2H}}$ in formula (15) is constant as $\eta\rightarrow0$, under the
condition of the Theorem. 
\end{remark}
Our numerical simulations indicate a Brownian behavior already for finite
$\eta$, but it cannot be strictly true, since - in spite of the explanation
given in Appendix \ref{sec:sigma_H_formula} - the process "feels" the presence of all the finite
number of fractional processes all the time, hence there is certainly a
residual of memory, although numerically very small. In the limit when
$\eta\rightarrow0$ the number of "vortex structures" of the noise goes to
infinity, the process "jumps" for one to the other and the approximate
property of independent increments observed for finite $\eta$ may become strict.

We do not know whether the previous theorem is true or not and hope its
statement will trigger research on it. For the time being, we offer the
following very partial numerical verification. 

Let us introduce the quantity
\[
\sigma_{t}^{2}\left(  \eta,\tau,H\right)  =\frac{\mathbb{E}\left[  \left\vert
\mathbf{X}_{t}^{\mathbf{0}}\right\vert ^{2}\right]  }{t},\qquad\text{for }t>0
\]
and its oscillation on a generic interval $\left[  t_{0},t_{1}\right]
\subset\left(  0,\infty\right)  $
\[
\Delta\left(  t_{0},t_{1},\eta,\tau,H\right)  =\sup_{t\in\left[  t_{0}%
,t_{1}\right]  }\sigma_{t}^{2}\left(  \eta,\tau,H\right)  -\inf_{t\in\left[
t_{0},t_{1}\right]  }\sigma_{t}^{2}\left(  \eta,\tau,H\right)  .
\]
We claim (in the numerical sense)

\begin{claim}
For $\tau=\tau_{\eta}=\eta^{\frac{1-2H}{1-H}}$, for every $\left[  t_{0}%
,t_{1}\right]  \subset\left(  0,\infty\right)  $,
\begin{equation}
\lim_{\eta\rightarrow0}\Delta\left(  t_{0},t_{1},\eta,\tau_{\eta},H\right)
=0.
\label{eq:main_claim}
\end{equation}
\end{claim}
In Appendix \ref{sec:sigma_H_formula} we add further discussion to this problem.

\section{Numerical results\label{sec:numerical}}
In this section we perform numerical simulations of transport equation (\ref{eq:system}) with advection velocity $\mathbf{u}(\mathbf{x},t)$ given by the stochastic model (\ref{eq:model}). A Monte Carlo method is employed where particle trajectories are simulated by numerical integration of the equations of characteristics (\ref{SDE}). Expected values and probabilities are thus approximated by appropriate ensemble averages. To illustrate the physics captured by our stochastic model we consider $H=0.7$. A simple explicit Euler method is used in all simulations to discretise time. The converge properties of Euler's method in the range $H>0.5$ can be found at \cite{Nualart}. 

As a preliminary step, we validate our numerical code against the theoretical predictions of the control case (see Appendix \ref{app_sec_control}). The velocity $u$ and the relaxation time $\tau_{\eta}$ are set to $1$ and $10^{-2}$, respectively. In the left panel of Fig. \ref{fig:results_control}  the probability $\mathbb{P}[|x(t)-x(0)| < R]$ is computed over $10^4$ realisations as a function of time (solid line) and compared with formula (\ref{eq:prob_th_control}) represented by the dashed line. Evidently, for $t \gg \tau_{\eta}$ the theoretical prediction and the numerical result overlap. In the right panel of Fig. \ref{fig:results_control} we report $\mathbb{E}[|x(t)-x(0)|^2]$ as a function of time for the same test case. Analogously, the numerical values (solid line) match the exact formula (dashed line) for $t \gg \tau_{\eta}$ as predicted. 

\begin{figure}[hbt!]
\centering
\begin{subfigure}[b]{0.45\textwidth}
\includegraphics[width=1\textwidth]{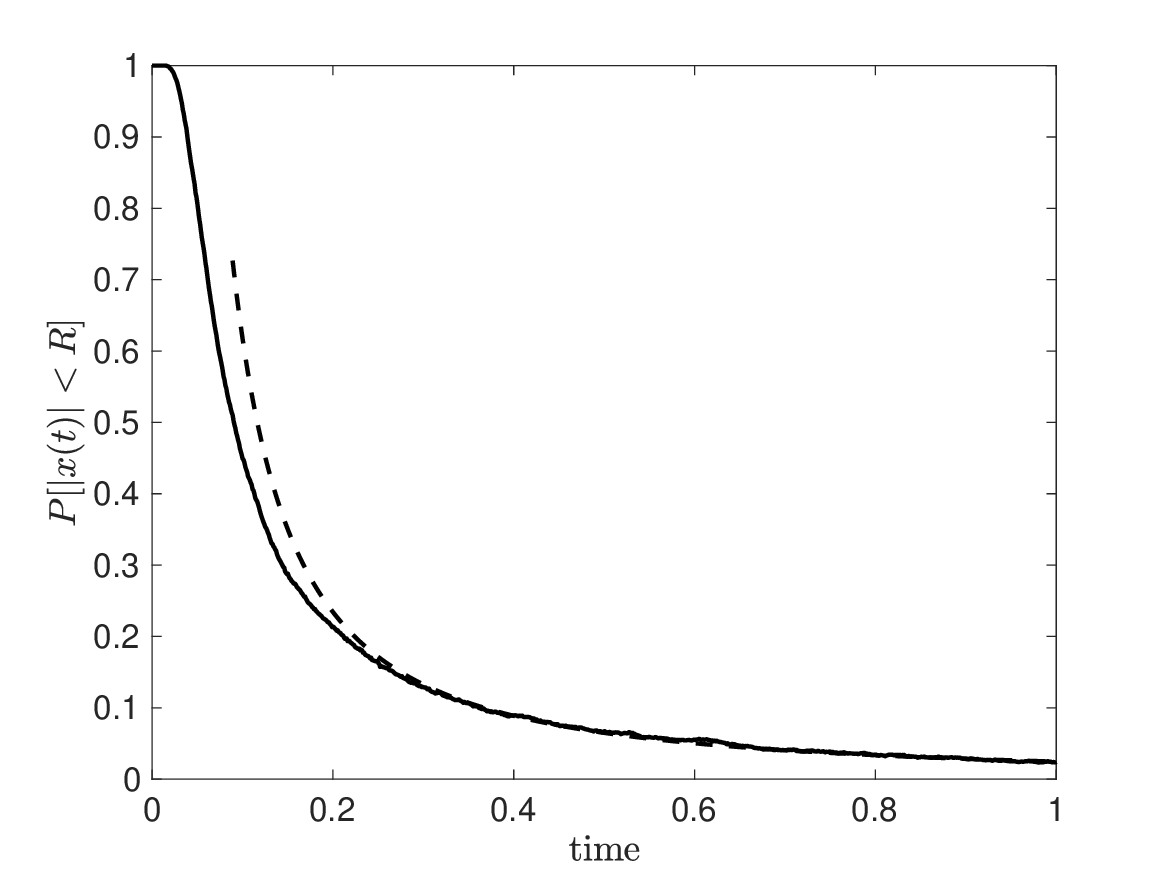}
\end{subfigure}
\hfill
\begin{subfigure}[b]{0.45\textwidth}
\includegraphics[width=\textwidth]{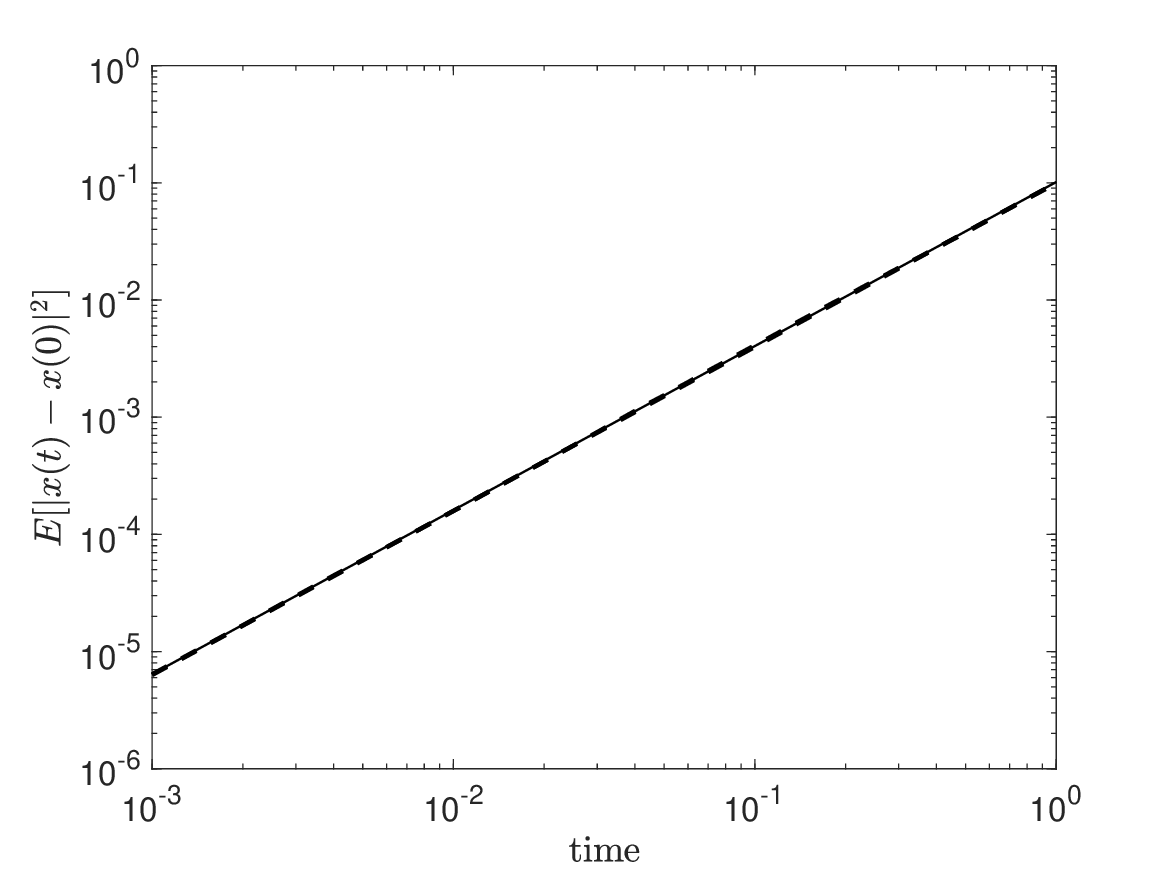}
\end{subfigure}
\caption{Probability $\mathbb{P}[|x(t)-x(0)| < R]$ (left panel) and variance $\mathbb{E}[|x(t)-x(0)|^2]$ (right panel) for the control case as a function of time. Numerical values are represented by the solid lines while the exact formulae are represented by the dashed lines.}
\label{fig:results_control}
\end{figure}

Having validated our numerical code, we move on to simulate non-trivial vector fields $\mathbf{\sigma}_{\mathbf{k}}$. In particular, we consider the vector field defined by (\ref{eq:sigma_test_case}), which represents a random perturbation concentrated at a length scale $\eta$. Having set $H=0.7$ we expect a particle transported by such $\mathbf{u}(\mathbf{x},t)$ to feel the ``memory effect'' due to correlated Brownian increments. An interesting question is then for how long this physical mechanism is maintained along a particle trajectory. To this aim we consider three values of $\eta=2\pi/20,2\pi/100,2\pi/200$. Reference velocity $u$ and relaxation time $\tau$ are set to $2$ and $10^{-2}$, respectively. Statistics are collected over $10^4$ realisations. Fig. \ref{fig:var_test_case} shows the variance $\mathbb{E}[|x(t)-x(0)|^2]$ as a function of time for the simulated values of $\eta$. 
\begin{figure}[hbt!]
\centering
\begin{subfigure}[b]{0.45\textwidth}
\includegraphics[width=1\textwidth]{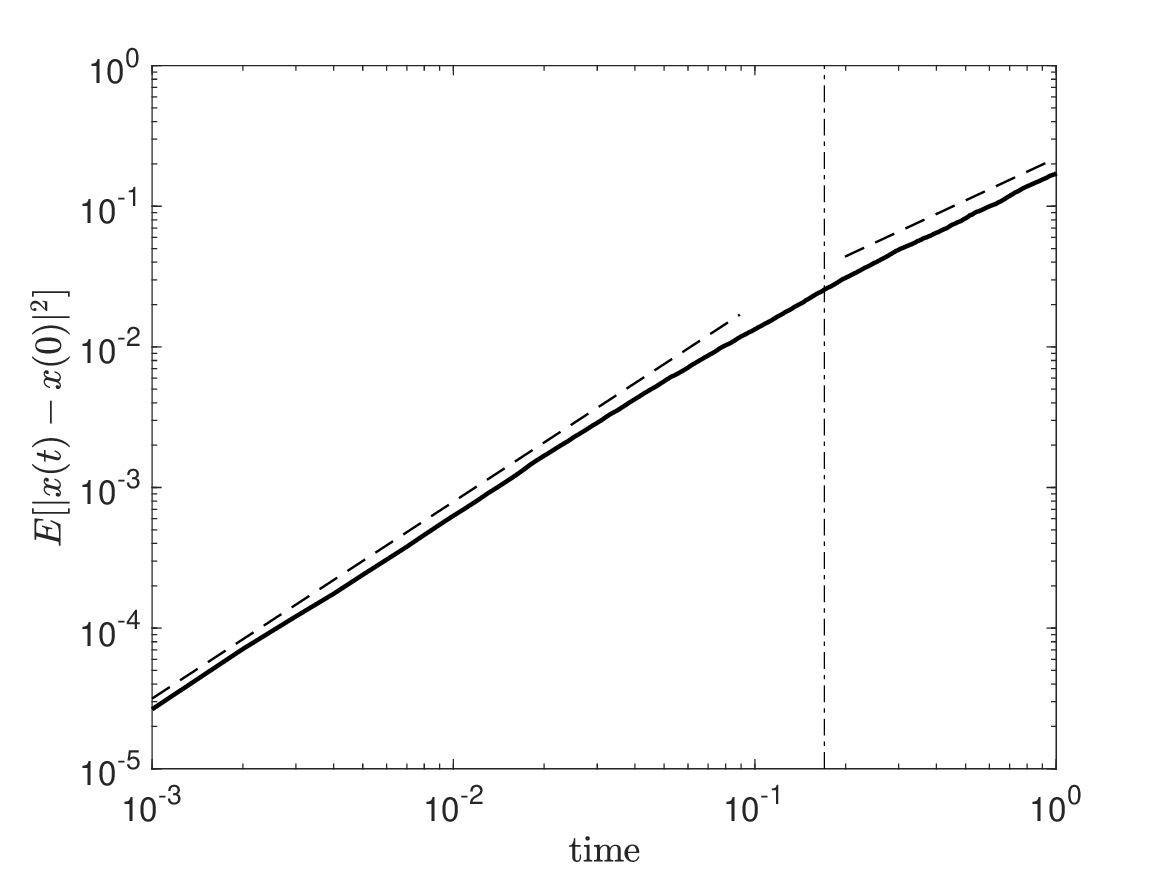}
\end{subfigure}
\hfill
\begin{subfigure}[b]{0.45\textwidth}
\includegraphics[width=\textwidth]{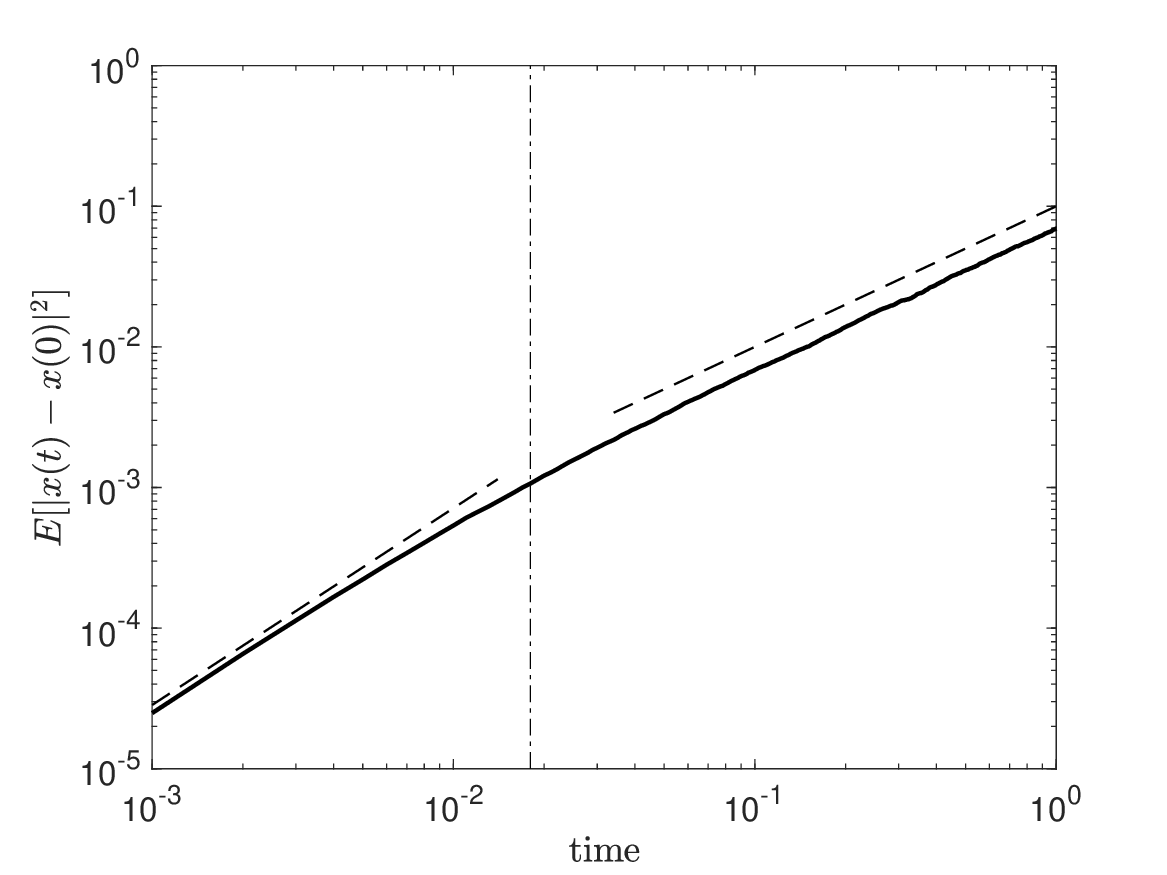}
\end{subfigure}
\hfill
\begin{subfigure}[b]{0.45\textwidth}
\includegraphics[width=\textwidth]{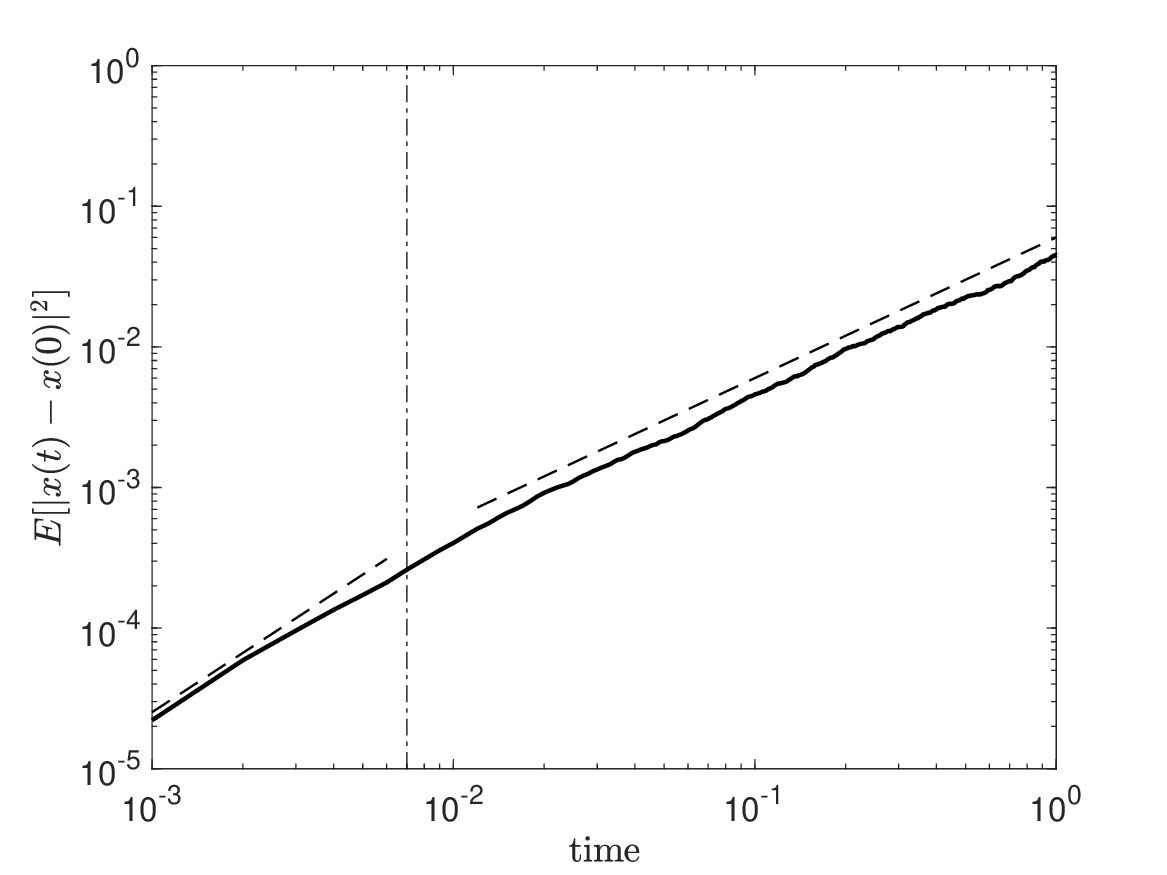}
\end{subfigure}
\caption{Variance $\mathbb{E}[|x(t)-x(0)|^2]$ for $\eta=\pi/20$ (left panel), $\eta=\pi/100$ (right panel) and $\eta=\pi/200$ (bottom panel) as a function of time. The dashed lines represents the slopes $2H$ and $1$. The dash-dotted vertical line represents $t=t^*$.}
\label{fig:var_test_case}
\end{figure}
This numerical test clearly highlights the presence of two regimes. For small times the Fractional Brownian Noise is dominant as a slope equal to $2H$ in the variance highlights. For larger times the classical Brownian Motion prevail restoring the variance scaling to $1$. Moreover, the time scale $t^*$ of departure between the two slopes decreases with $\eta$ and its value depends on the choice of parameters $C(\eta,\tau,H)$ and $u$. For $t \ll t^*$ the vector fields $\mathbf{\sigma}_{\mathbf{k}}$ are approximately constant and the resulting motion is clearly an FBM. As time increases, i.e. the particle travels a distance of order $\eta$, the particle is selectively affected by independent components of the noise thus restoring Brownian Motion, as heuristically motivated in Sec. \ref{sec:main_res}. In the plots of Fig. \ref{fig:var_test_case}) $t^*$ (vertical lines in Fig. \ref{fig:var_test_case}) is set to be the time at which the particle has travelled approximately a distance of $\eta/2$. 

The standard deviation of the Brownian Motion established in the region for $t \gg t^*$, called $\sigma_H$, can be estimated by a diffusive approximation of the fractional Brownian Motion as detailed in Appendix \ref{sec:sigma_H_formula}. In Fig. \ref{fig:sigma_H} $\sigma_H$ computed from numerical simulation is shown as function of $\eta$ (dots) and compared against expression (\ref{eq:sigma_H}) (solid line). 
\begin{figure}[hbt!]
\centering
\includegraphics[width=0.6\textwidth]{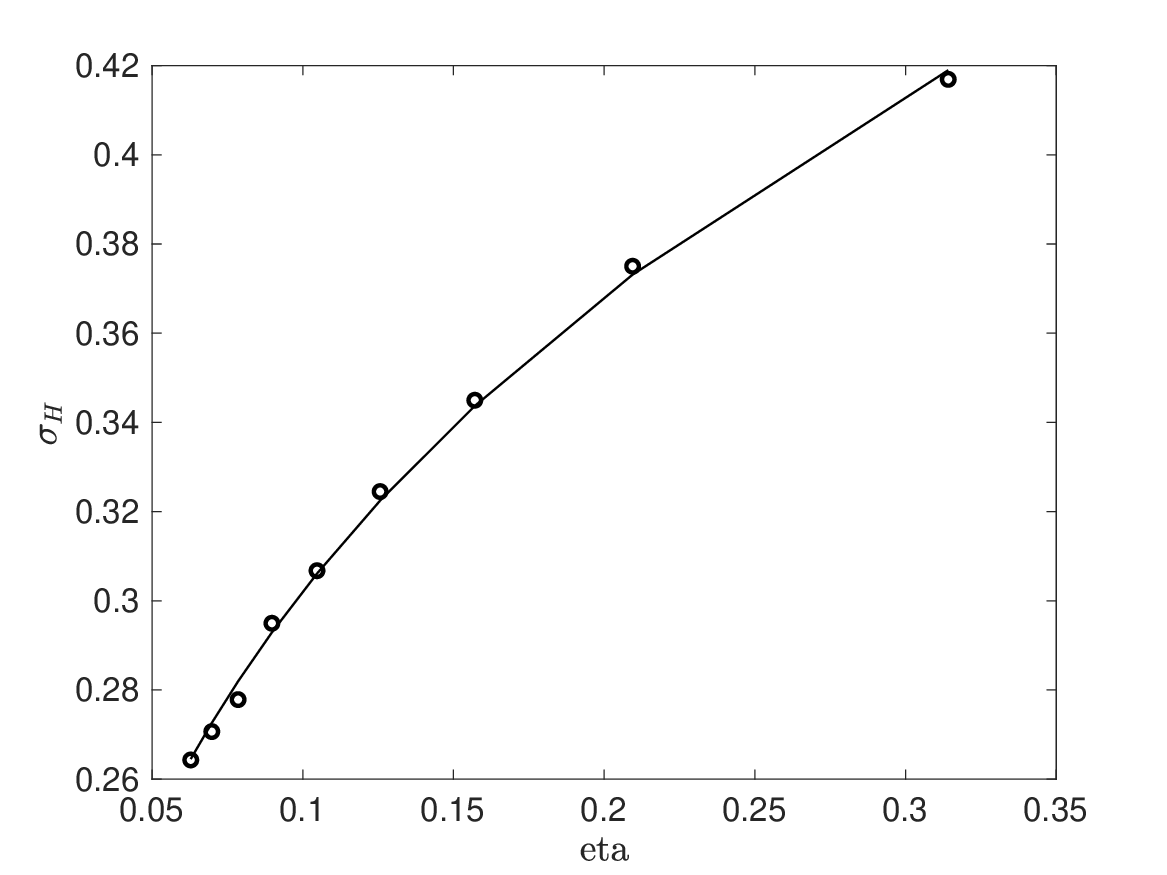}
\caption{Standard deviation $\sigma_H$ as a function of $\eta$ computed numerically (dots) and analytically (solid line) using (\ref{eq:sigma_H}).}
\label{fig:sigma_H}
\end{figure}
A good agreement of the theoretical prediction with the numerical results is found. The free parameter $\lambda$ computed using a least-squares method is approximately $0.47$. The latter is, however, not a universal constant but rather it has been found dependent on $\mathbf{\sigma}_{\mathbf{k}}$. To show this behaviour we generalise (\ref{eq:sigma_test_case}) to
\begin{equation}
\begin{aligned}
&\mathbf{\sigma}_{\mathbf{k}}\left(  \mathbf{x}\right)  =\frac{\mathbf{k}%
^{\perp}}{\left\vert \mathbf{k}\right\vert }g\left( \cos\left(  \mathbf{k}%
\cdot\mathbf{x}\right) \right)  \qquad\text{if }\mathbf{k\in K}_{\eta}^{+}\text{ }%
\\
&\mathbf{\sigma}_{\mathbf{k}}\left(  \mathbf{x}\right)  =\frac{\mathbf{k}%
^{\perp}}{\left\vert \mathbf{k}\right\vert } g\left( \sin\left(  \mathbf{k}%
\cdot\mathbf{x}\right) \right) \qquad\text{if }\mathbf{k\in K}_{\eta}^{-}%
\end{aligned}
\label{eq:sigma_test_case_tanh}
\end{equation}
where
\begin{equation*}
g(r) = \tanh(Mr),
\end{equation*}
is an approximation of the sign function with smoothness controlled by the parameter $M$. By repeating the above procedure we find $\lambda \approx 0.41$ for $M=10$. We thus conclude that there exists indeed a constant $\lambda$, but it is not universal. 

A concise way to state our results is formulated by Claim (\ref{eq:main_claim}). While proving this statement is challenging, we here limit ourselves to simulate the oscillation $\Delta(t_0,t_1,\eta,\tau,H)$ for decreasing values of $\eta$.
\begin{figure}[hbt!]
\centering
\includegraphics[width=0.6\textwidth]{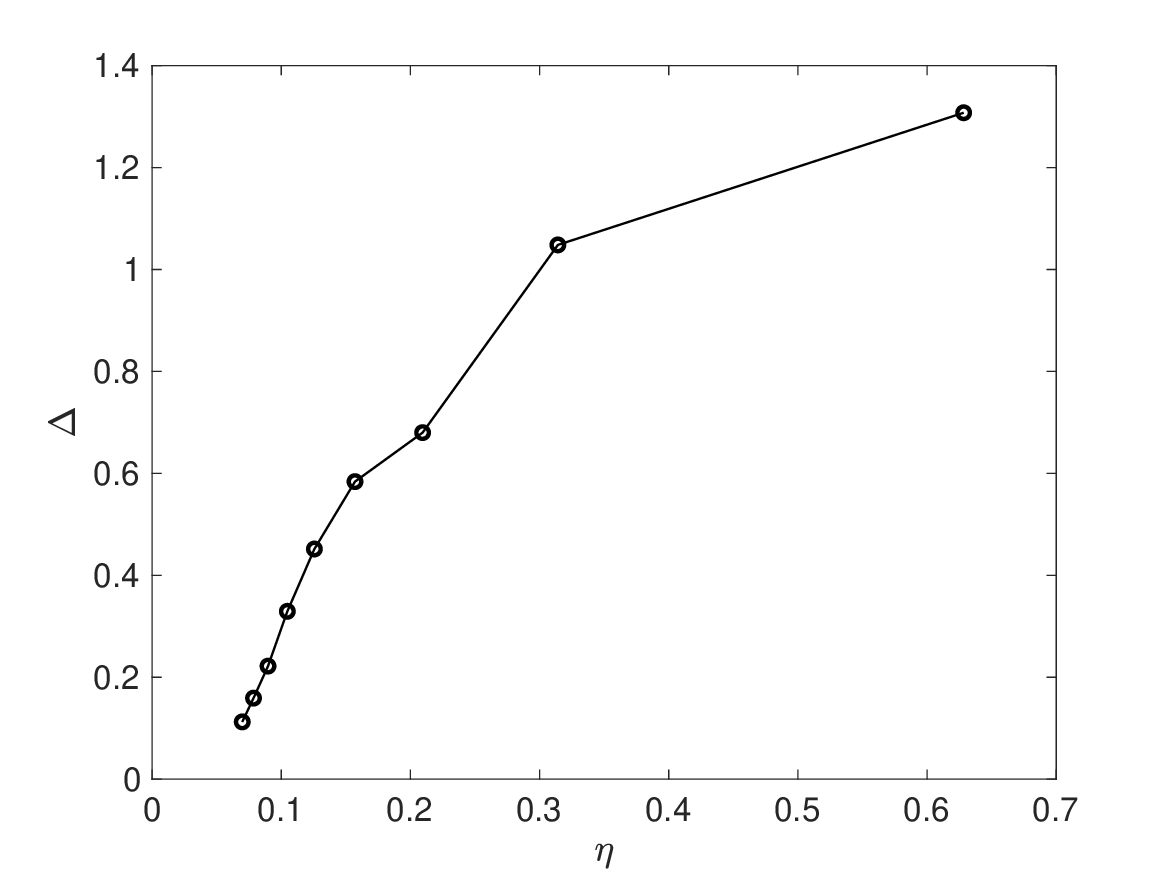}
\caption{Oscillation $\Delta(t_0,t_1,\eta,\tau,H)$ as a function of $\eta$ computed numerically.}
\label{fig:delta_claim}
\end{figure}
The findings, reported in Fig. \ref{fig:delta_claim}, support our claim with $\Delta$ approaching zero as $\eta$ decreases. We remark, however, that the limit $\eta \to 0$ is computationally not feasible. To smaller $\eta$ there correspond larger radius in spectral space therefore increasing the terms to be computed, at each time step, in the expression of $\mathbf{\sigma}_{\mathbf{k}}$. Furthermore, the velocity components of smaller length-scale $\eta$ have larger gradients, which require a finer time step to be properly resolved. This, in practice, has limited the numerical investigation conducted here to $\eta \approx 7\cdot 10^{-2}$. Nevertheless, the trend is indeed in agreement with our predictions. 

As a final and more involved illustration of our stochastic model we consider the evolution of an ensemble of $N=10^3$ particles. The underlying vector field is again given by (\ref{eq:sigma_test_case}) where we set $\eta=2\pi/20$. The particles are released at $t=0$ according to a uniform random distribution in a circle of radius $R=\eta/8$ centred at the origin. A number of $10^3$ realisations is simulated and the variance $\text{VAR}_N(t)=\frac{1}{N}\sum_{i=1}^N \mathbb{E}[|\mathbf{x}_i(t)|^2]$ is computed over time. The qualitative behaviour of this swarm of particles is shown in Fig. \ref{fig:multip_times} where the simulated particle positions are drown at different times. For $t \ll t^*$ (top-right panel) the particles are pushed around the origin by an approximately spatially uniform velocity, thus maintaining an almost circular shape. For times of order $t^*$ (mid-left and mid-right panel) particles tend to agglomerate in regions where the vector field components $\sigma_\mathbf{k}$ sum to zero, which are characterised by a typical length scale of order $\eta$. Finally, for $t \gg t^*$ (bottom-right and bottom-left panel) the particles are picked up by different and independent waves composing the noise and thus loosing all structures inserted in the initial conditions. 

\begin{figure}[hbt!]
\centering
\begin{subfigure}[b]{0.45\textwidth}
\includegraphics[width=1\textwidth]{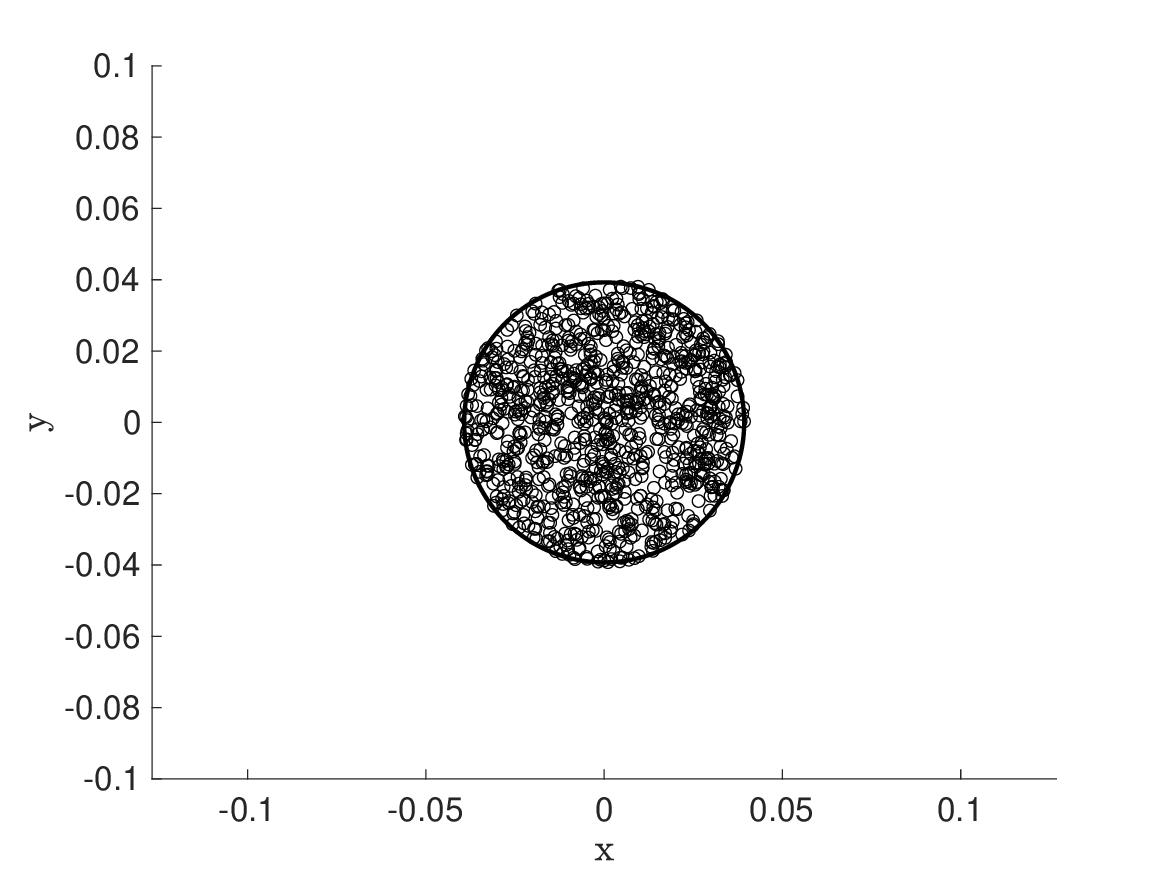}
\end{subfigure}
\hfill
\begin{subfigure}[b]{0.45\textwidth}
\includegraphics[width=\textwidth]{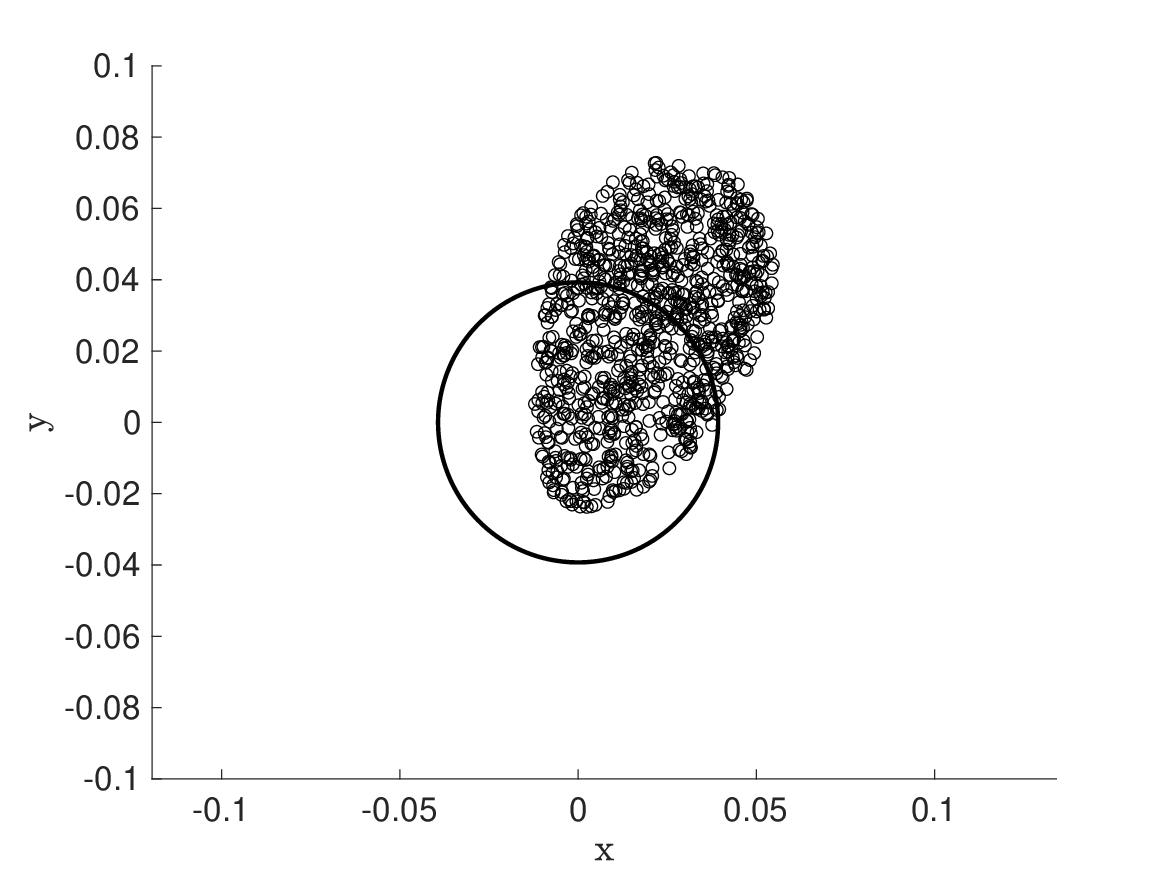}
\end{subfigure}
\hfill
\begin{subfigure}[b]{0.45\textwidth}
\includegraphics[width=\textwidth]{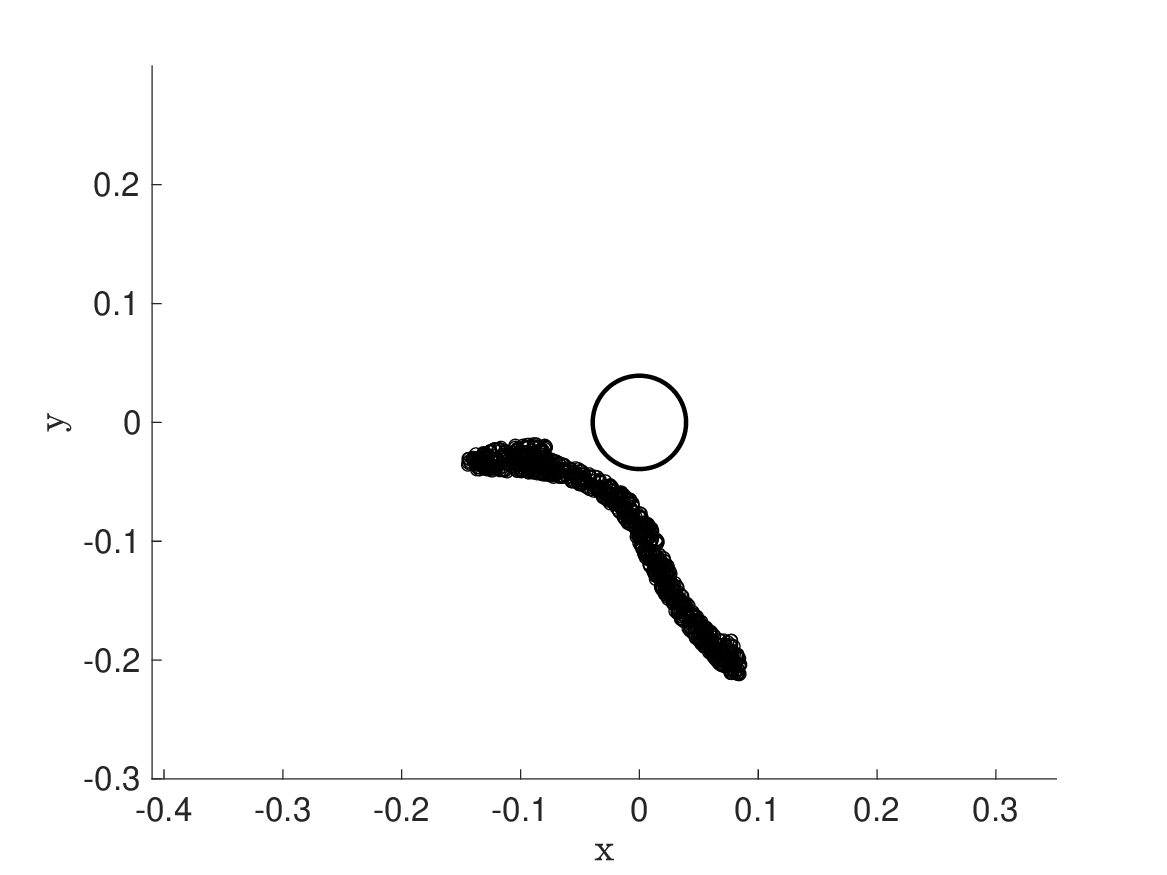}
\end{subfigure}
\hfill
\begin{subfigure}[b]{0.45\textwidth}
\includegraphics[width=\textwidth]{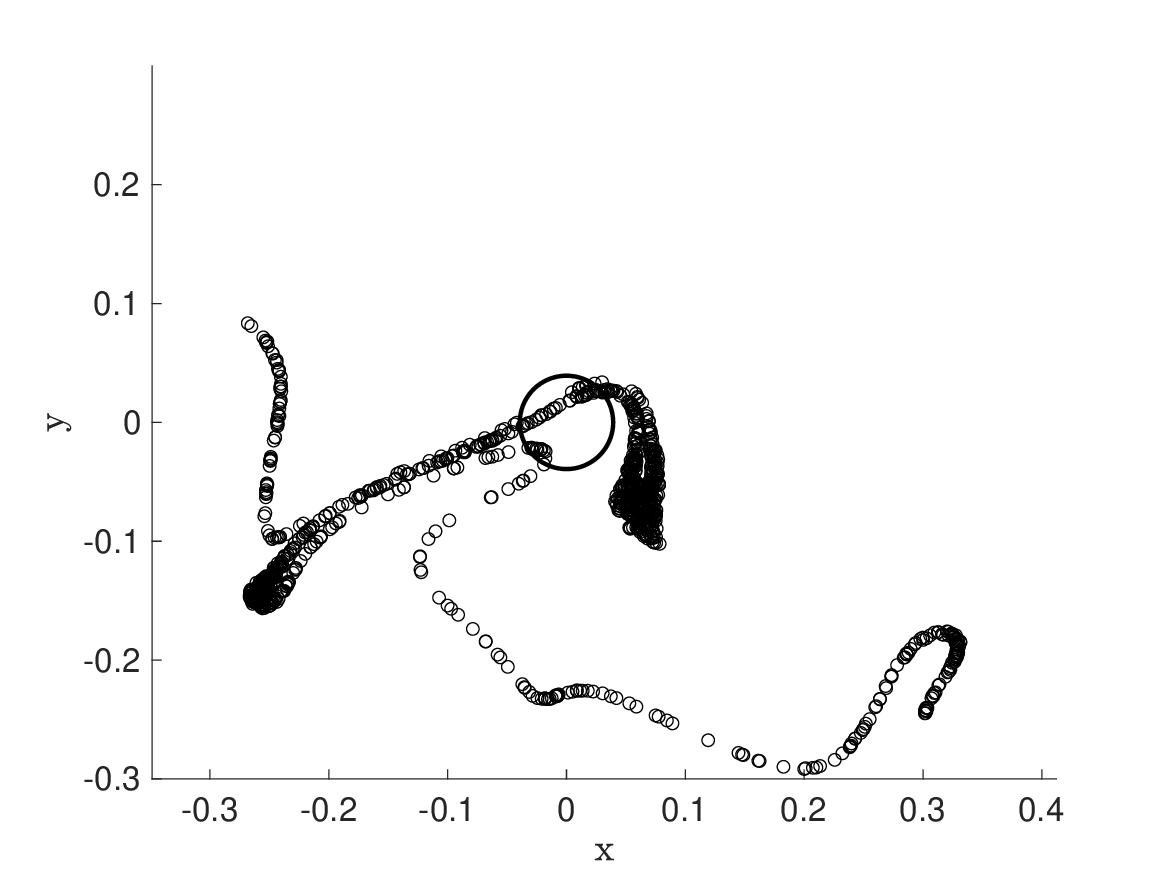}
\end{subfigure}
\hfill
\begin{subfigure}[b]{0.45\textwidth}
\includegraphics[width=\textwidth]{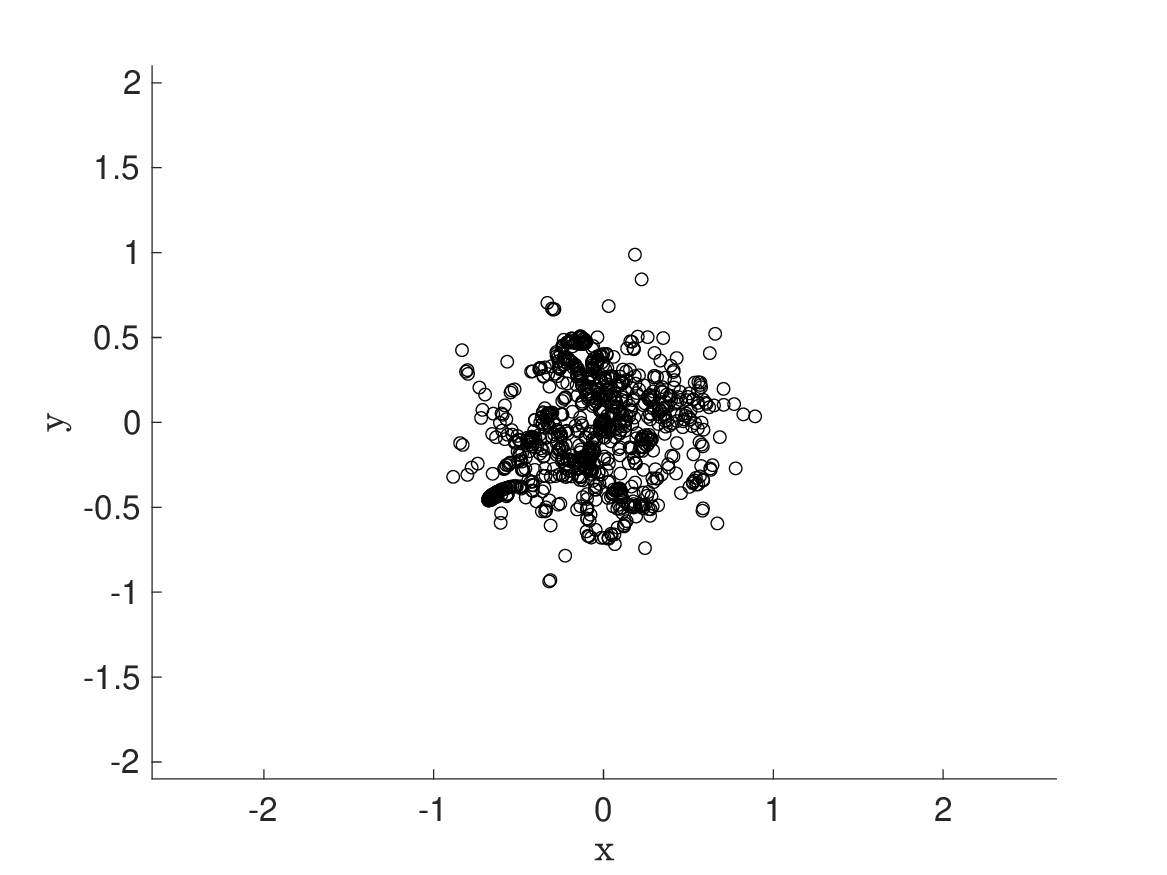}
\end{subfigure}
\hfill
\begin{subfigure}[b]{0.45\textwidth}
\includegraphics[width=\textwidth]{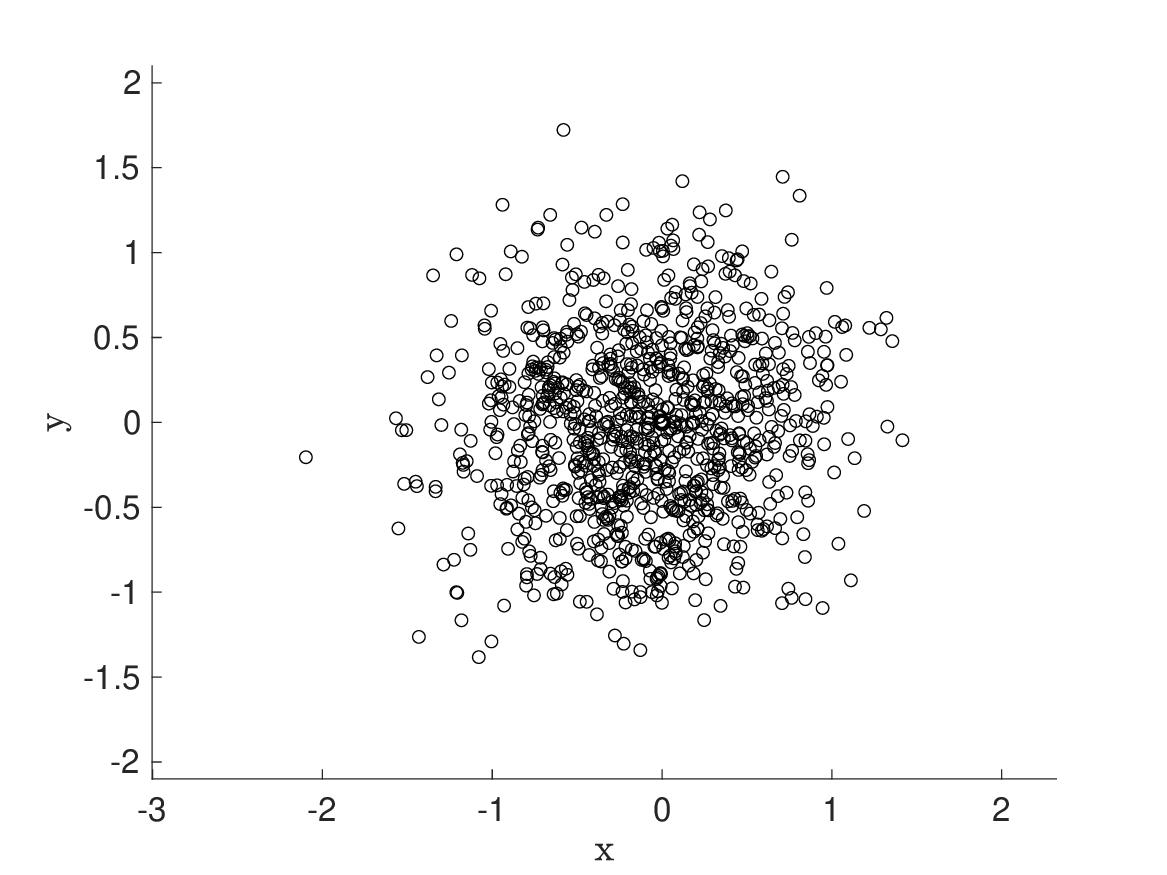}
\end{subfigure}
\caption{Particle positions at $t=0$ (top-left panel), $t=t^*/10$ (top-right panel), $t=t^*$ (mid-left panel), $t=4t^*$ (mid-right panel), $t=10t^*$ (bottom-left panel) and $t=25t^*$ (bottom-right panel). The circle $r=\eta/8$ is depicted by solid line in all figures as point of reference.}
\label{fig:multip_times}
\end{figure}

Fig. \ref{fig:var_N} shows VAR$_N(t)$ for $H=0.7$ (solid line) and $H=0.5$ (dash-dotted line). The Hurst exponent does not appear to have a significant influence for $t \ll t^*$ where the collective behaviour of the particles is dominated by the particular choice of vector field components $\sigma_\textbf{k}$. As time becomes much larger than $t^*$ a Brownian Motion dispersion is recovered. However, to a larger $H$ there corresponds a larger variance as suggested by Remark 6. of Appendix \ref{sec:sigma_H_formula}. 

\begin{figure}[hbt!]
\centering
\includegraphics[width=0.6\textwidth]{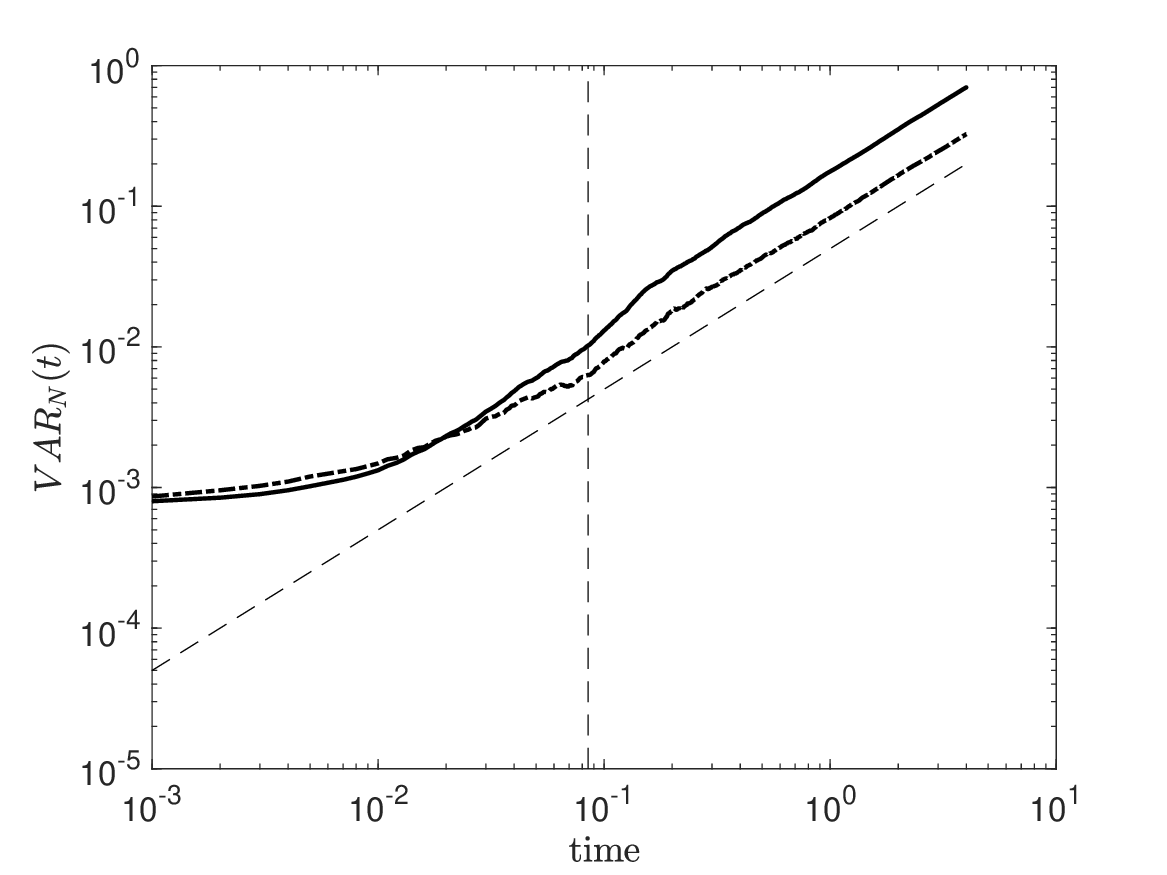}
\caption{Particles variance VAR$_N(t)$ as a function of time for $H=0.7$ (solid line) and $H=0.5$ (dash-dotted line). Slope $1$ is represented by the inclined dashed line while $t=t^*$ is represented by the vertical dashed line.}
\label{fig:var_N}
\end{figure}

\section{Conclusions\label{sec:concl}}
We have investigated the dispersion of particles transported by a stochastic vector field driven by fractional gaussian processes at different scales. We found the existence of two regimes: at small times compared to those taken by a particle to travel a distance of order $\eta$, the motion is governed by the fractional gaussian process of the single vector field components. As the particle travels distances larger than $\eta$, a standard Brownian Motion is recovered. This behaviour is indicated by the double slope in the variance found in the numerical simulations, ranging from $H>1/2$ in the first regime to $H=1/2$ in the second regime. We derived, by a diffusion approximation, an expression for the standard deviation of the Brownian Motion established in the second regime and verified the formula against numerical predictions. A good agreement was found up to an arbitrary constant $\lambda$. 

We then postulated Theorem \ref{theo_main} where converges in law to a 2-dimensional Brownian is expected in the limit $\eta \to 0$. No attempt to prove this theorem was made in this work, but rather we empirically tested its validity through numerical simulations. Our findings, even though computationally  limited to a finite $\eta$, are consistent to the statement of Theorem \ref{theo_main}. We hope that this work will spark interest in the scientific community and serve as a base ground to build upon further studies on Fractional Brownian Motion.

\appendix

\section{Model scaling\label{appendix noise structure}}

Consider the velocity field $\mathbf{u}_{\tau}\left(  \mathbf{x},t\right)  $
given by equation (\ref{smooth velocity}), where $Z_{t}^{\tau,H,\mathbf{k}}$
is given by (\ref{OU}) and $B_{t}^{H,\mathbf{k}}$ are independent FBM's.

\begin{remark}
In law, $Z_{t}^{\tau,H,\mathbf{k}}=\theta_{t/\tau}^{H,\mathbf{k}}$ where
$\theta_{t}^{H,\mathbf{k}}$ satisfies
\[
d\theta_{t}^{H,\mathbf{k}}=-\theta_{t}^{H,\mathbf{k}}dt+\frac{\sqrt{2}}%
{\sqrt{\Gamma\left(  2H+1\right)  }}dB_{t}^{H,\mathbf{k}}.
\]
The proof is based on the fact that
\[
\frac{1}{\tau}\frac{dB_{\cdot}^{H,\mathbf{k}}}{dt}|_{\frac{t}{\tau}}=\frac
{d}{dt}B_{t/\tau}^{H,\mathbf{k}}=\tau^{-H}\frac{d}{dt}\tau^{H}B_{t/\tau
}^{H,\mathbf{k}}=\tau^{-H}\frac{d}{dt}\widetilde{B}_{t}^{H,\mathbf{k}}%
\]
where $\widetilde{B}_{t}^{H,\mathbf{k}}$ is another FBM, since $\tau
^{H}B_{t/\tau}^{H,\mathbf{k}}\overset{\mathcal{L}}{=}\widetilde{B}%
_{t}^{H,\mathbf{k}}$.
\end{remark}

We have:

\begin{lemma}\label{Lemma_norm}
$\mathbb{E}\left[  Z_{t}^{\tau,H,\mathbf{k}}\right]  =0$,
\begin{equation}
\lim_{t \to \infty} \mathbb{E}\left[  \left(  Z_{t}^{\tau,H,\mathbf{k}}\right)  ^{2}\right]
=1\label{varianza uno}%
\end{equation}
\end{lemma}

\begin{proof}
We have 
\begin{equation*}
Z_{t}^{\tau ,H,\mathbf{k}}=\int_{0}^{t}e^{-\frac{t-s}{\tau }}\frac{\sqrt{2}}{%
\sqrt{\Gamma \left( 2H+1\right) }\tau ^{H}}dB_{s}^{H,\mathbf{k}}.
\end{equation*}%
From \cite{Nualart}, Chapter 5, Section 1, we have%
\begin{eqnarray*}
\mathbb{E}\left[ \left( \int_{0}^{t}e^{-\frac{t-s}{\tau }}dB_{s}^{H,\mathbf{k%
}}\right) ^{2}\right]  &=&H\left( 2H-1\right) \int_{0}^{t}\int_{0}^{t}e^{-%
\frac{t-s}{\tau }}e^{-\frac{t-r}{\tau }}\left\vert s-r\right\vert ^{2H-2}dsdr
\\
&=&H\left( 2H-1\right) \int_{0}^{t}\int_{0}^{t}e^{-\frac{s}{\tau }}e^{-\frac{%
r}{\tau }}\left\vert s-r\right\vert ^{2H-2}dsdr.
\end{eqnarray*}%
Therefore 
\begin{eqnarray*}
\lim_{t\rightarrow \infty }\mathbb{E}\left[ \left( Z_{t}^{\tau ,H,\mathbf{k}%
}\right) ^{2}\right]  &=&\frac{2}{\Gamma \left( 2H+1\right) \tau ^{2H}}%
H\left( 2H-1\right) \int_{0}^{\infty }\int_{0}^{\infty }e^{-\frac{s}{\tau }%
}e^{-\frac{r}{\tau }}\left\vert s-r\right\vert ^{2H-2}dsdr \\
&=&\frac{2H\left( 2H-1\right) }{\Gamma \left( 2H+1\right) }\int_{0}^{\infty
}\int_{0}^{\infty }e^{-s^{\prime }}e^{-r^{\prime }}\left\vert s^{\prime
}-r^{\prime }\right\vert ^{2H-2}ds^{\prime }dr^{\prime } \\
&=&1
\end{eqnarray*}%
because%
\begin{eqnarray*}
&&\int_{0}^{\infty }\int_{0}^{\infty }e^{-s}e^{-r}\left\vert s-r\right\vert
^{2H-2}dsdr \\
&=&2\int_{0}^{\infty }e^{-s}\left( \int_{0}^{s}e^{-r}\left\vert
s-r\right\vert ^{2H-2}dr\right) ds \\
&=&2\int_{0}^{\infty }e^{-2s}\left( \int_{0}^{s}e^{r}r^{2H-2}dr\right) ds \\
&=&2\left[ -\frac{1}{2}e^{-2s}\int_{0}^{s}e^{r}r^{2H-2}dr\right]
_{s=0}^{s=\infty }+2\int_{0}^{\infty }\frac{1}{2}e^{-2s}e^{s}s^{2H-2}ds \\
&=&\int_{0}^{\infty }e^{-s}s^{2H-2}ds=\Gamma \left( 2H-1\right) =\frac{%
\Gamma \left( 2H+1\right) }{2H\left( 2H-1\right) }
\end{eqnarray*}%
recalling that $\Gamma \left( z+1\right) =\int_{0}^{\infty }e^{-s}s^{z}ds$
and $\Gamma \left( z+1\right) =z\Gamma \left( z\right) $.
\end{proof}

From the result of Lemma \ref{Lemma_norm} and the independence, %

\[
\mathbb{E}\left[  \left\vert \mathbf{u}_{\tau}\left(  \mathbf{x},t\right)
\right\vert ^{2}\right]  =\frac{u^{2}}{C_{\eta}^{2}}\sum_{\mathbf{k}%
\in\mathbf{K}_{\eta}}\left\vert \mathbf{\sigma}_{\mathbf{k}}\left(
\mathbf{x}\right)  \right\vert ^{2}.
\]
Hence (recall (\ref{condition on sigma})-(\ref{def of C eta})) the limit
\[
\left\langle \mathbf{u}_{\tau}\right\rangle ^{2}:=\lim_{R\rightarrow\infty
}\frac{1}{R^{2}}\int_{\left[  -\frac{R}{2},\frac{R}{2}\right]  ^{2}}%
\mathbb{E}\left[  \left\vert \mathbf{u}_{\tau}\left(  \mathbf{x},t\right)
\right\vert ^{2}\right]  dx
\]
exists and is given by%
\[
\left\langle \mathbf{u}_{\tau}\right\rangle ^{2}=u^{2}.
\]
This is the motivation for calling $u^{2}$ the mean square turbulent velocity
(or turbulent kinetic energy, multiplied by 2). Moreover, the velocity fields
$\mathbf{u}_{\tau}\left(  \mathbf{\cdot},t\right)  $ are
\textit{instantaneously} on the same ground with respect to the parameters
$\eta,\tau,H$, namely at a given time they have the same average intensity.

Since%
\[
\tau dZ_{t}^{\tau,H,\mathbf{k}}=-Z_{t}^{\tau,H,\mathbf{k}}dt+\tau^{1-H}%
\frac{\sqrt{2}}{\sqrt{\Gamma\left(  2H+1\right)  }}dB_{t}^{H,\mathbf{k}}%
\]
we may conjecture that
\[
Z_{t}^{\tau,H,\mathbf{k}}dt\sim\tau^{1-H}\frac{\sqrt{2}}{\sqrt{\Gamma\left(
2H+1\right)  }}dB_{t}^{H,\mathbf{k}}.
\]
This fact has a rigorous formulation as shown in the next lemma. On its basis, we replace the velocity
field $\mathbf{u}_{\tau}\left(  \mathbf{x},t\right)  $ above by
\[
\mathbf{u}\left(  \mathbf{x},t\right)  =\frac{u\tau^{1-H}\sqrt{2}}%
{\sqrt{\Gamma\left(  2H+1\right)  }C_{\eta}}\sum_{\mathbf{k}\in\mathbf{K}%
_{\eta}}\mathbf{\sigma}_{\mathbf{k}}\left(  \mathbf{x}\right)  \frac
{dB_{t}^{H,\mathbf{k}}}{dt}.
\]

\begin{lemma}\label{Lemma_main}
\begin{equation*}
\mathbb{E}\left[ \left( \int_{0}^{t}Z_{s}^{\tau ,H,\mathbf{k}}ds-\tau ^{1-H}%
\frac{\sqrt{2}}{\sqrt{\Gamma \left( 2H+1\right) }}B_{t}^{H,\mathbf{k}%
}\right) ^{2}\right] \leq \tau ^{2-2H}.
\end{equation*}
\end{lemma}

\begin{proof}
As a preliminary remark, we notice that one can use Fubini theorem also for
the Wiener integrals with respect to fractional Brownian motion, thanks to
the reformulation%
\begin{equation*}
\int_{0}^{t}e^{-\frac{\left( t-s\right) }{\tau }}dB_{s}^{H,\mathbf{k}}=-%
\frac{1}{\tau }\int_{0}^{t}e^{-\frac{\left( t-s\right) }{\tau }}B_{s}^{H,%
\mathbf{k}}ds+B_{t}^{H,\mathbf{k}}
\end{equation*}%
and the application of Fubini theorem to the classical integral on the
right-hand-side of this identity. Based on this preliminary fact, we have%
\begin{eqnarray*}
\int_{0}^{t}Z_{s}^{\tau ,H,\mathbf{k}}ds &=&\frac{\sqrt{2}}{\sqrt{\Gamma
\left( 2H+1\right) }\tau ^{H}}\int_{0}^{t}\left( \int_{0}^{s}e^{-\frac{%
\left( s-r\right) }{\tau }}dB_{r}^{H,\mathbf{k}}\right) ds \\
&=&\tau ^{1-H}\frac{\sqrt{2}}{\sqrt{\Gamma \left( 2H+1\right) }}%
\int_{0}^{t}\left( \int_{r}^{t}\frac{1}{\tau }e^{-\frac{\left( s-r\right) }{%
\tau }}ds\right) dB_{r}^{H,\mathbf{k}} \\
&=&\tau ^{1-H}\frac{\sqrt{2}}{\sqrt{\Gamma \left( 2H+1\right) }}B_{t}^{H,%
\mathbf{k}}+R_{t}^{H,\mathbf{k}}
\end{eqnarray*}%
having used $\int_{r}^{t}\frac{1}{\tau }e^{-\frac{\left( s-r\right) }{\tau }%
}ds=1-e^{-\frac{\left( t-r\right) }{\tau }}$, where we set 
\begin{equation*}
R_{t}^{H,\mathbf{k}}=\tau ^{1-H}\frac{\sqrt{2}}{\sqrt{\Gamma \left(
2H+1\right) }}\int_{0}^{t}e^{-\frac{\left( t-r\right) }{\tau }}dB_{r}^{H,%
\mathbf{k}}.
\end{equation*}%
Then it is sufficient to prove that 
\begin{equation*}
\mathbb{E}\left[ \left( R_{t}^{H,\mathbf{k}}\right) ^{2}\right] \leq \tau
^{2-2H}.
\end{equation*}%
But this follows immediately from the result of the previous lemma.
\end{proof}

\section{On formula (\ref{Lagrangian formulation})}\label{sec:app_b}

In this appendix we discuss the formula (\ref{Lagrangian formulation}). If
$\phi:\mathbb{R}^{2}\rightarrow\mathbb{R}$ is a measurable compact support
test function, then
\[
\left\langle T\left(  t\right)  ,\phi\right\rangle =\int\phi\left(
\mathbf{y}\right)  T\left(  \mathbf{y},t\right)  d^{2}y=\int\phi\left(
\mathbf{X}_{t}^{\mathbf{x}}\right)  T\left(  \mathbf{X}_{t}^{\mathbf{x}%
},t\right)  d^{2}x=\int\phi\left(  \mathbf{X}_{t}^{\mathbf{x}}\right)
T_{0}\left(  \mathbf{x}\right)  d^{2}x
\]
where the intermediate identity, based on the change of variable
$\mathbf{y}=\mathbf{X}_{t}^{\mathbf{x}}$, is due to the fact that the
determinant of the Jacobian of $\mathbf{x}\rightarrow\mathbf{X}_{t}%
^{\mathbf{x}}$ solves an equation with the trace of the derivative of the
$\mathbf{\sigma}_{\mathbf{k}}\left(  \mathbf{x}\right)  $, which is the
divergence of $\mathbf{\sigma}_{\mathbf{k}}\left(  \mathbf{x}\right)  $, hence
equal to zero; therefore the Jacobian determinant is equal to one.

In case one is interested in single realizations of $\left\langle T\left(
t\right)  ,\phi\right\rangle $, a numerical method is the following one:
generate a sample of $N$ points $\mathbf{x}_{i}$, $i=1,...,N$, distributed
according to the density $T_{0}\left(  \mathbf{x}\right)  $, hence compute%
\[
\left\langle T\left(  t\right)  ,\phi\right\rangle \sim\frac{1}{N}\sum
_{i=1}^{N}\phi\left(  \mathbf{X}_{t}^{\mathbf{x}_{i}}\right)  .
\]
However, here we mean that we use the same noise realizations for each one of
the points $\mathbf{x}_{i}$. 

In this paper we want to investigate a number of quantities, depending on the
Hurst exponent $H$ and the set $\mathbf{K}_{\eta}$ (and $\phi$):

\begin{enumerate}
\item the function $t\mapsto\mathbb{E}\left[  \left\vert \mathbf{X}%
_{t}^{\mathbf{0}}\right\vert ^{2}\right]  $

\item the mean value $t\mapsto m_{t}^{\phi}=\mathbb{E}\left[  \left\langle
T\left(  t\right)  ,\phi\right\rangle \right]  $

\item the variance $t\mapsto\mathbb{E}\left[  \left(  \left\langle T\left(
t\right)  ,\phi\right\rangle -m_{t}^{\phi}\right)  ^{2}\right]  $.
\end{enumerate}

\subsection{Particular choice of $T_{0}$ and $\phi$}

Let us make the following special choices:\
\[
T_{0}=\delta_{0}%
\]
namely the weak limit of densities of the form $T_{0}^{\epsilon}\left(
\mathbf{x}\right)  =\epsilon^{-2}\theta\left(  \epsilon^{-1}\mathbf{x}\right)
$ with suitable pdf $\theta$, and
\[
\phi=1_{B\left(  \mathbf{0},R\right)  }%
\]
the indicator function of the ball $B\left(  \mathbf{0},R\right)  $. We
expect, in the average, a decrease of $\left\langle T\left(  t\right)
,\phi\right\rangle $.

In this case, using the formula $\left\langle T\left(  t\right)
,\phi\right\rangle =\int\phi\left(  \mathbf{X}_{t}^{\mathbf{x}}\right)
T_{0}\left(  \mathbf{x}\right)  d^{2}x$ and the approximation of $T_{0}%
=\delta_{0}$ by $T_{0}^{\epsilon}$, we get%
\[
\left\langle T\left(  t\right)  ,\phi\right\rangle =\phi\left(  \mathbf{X}%
_{t}^{\mathbf{0}}\right)  =\left\{
\begin{array}
[c]{ccc}%
1 & \text{if} & \left\vert \mathbf{X}_{t}^{\mathbf{0}}\right\vert <R\\
0 & \text{if} & \left\vert \mathbf{X}_{t}^{\mathbf{0}}\right\vert \geq R
\end{array}
\right.
\]
and therefore%
\[
\mathbb{E}\left[  \left\langle T\left(  t\right)  ,\phi\right\rangle \right]
=\mathbb{P}\left(  \left\vert \mathbf{X}_{t}^{\mathbf{0}}\right\vert
<R\right)
\]%
\[
\mathbb{E}\left[  \left(  \left\langle T\left(  t\right)  ,\phi\right\rangle
-m_{t}^{\phi}\right)  ^{2}\right]  =\mathbb{E}\left[  \phi\left(
\mathbf{X}_{t}^{\mathbf{0}}\right)  ^{2}\right]  -\mathbb{P}\left(  \left\vert
\mathbf{X}_{t}^{\mathbf{0}}\right\vert <R\right)  ^{2}=\mathbb{P}\left(
\left\vert \mathbf{X}_{t}^{\mathbf{0}}\right\vert <R\right)  -\mathbb{P}%
\left(  \left\vert \mathbf{X}_{t}^{\mathbf{0}}\right\vert <R\right)  ^{2}%
\]
(because $\phi\left(  \mathbf{X}_{t}^{\mathbf{0}}\right)  ^{2}=\phi\left(
\mathbf{X}_{t}^{\mathbf{0}}\right)  $)
\[
=\mathbb{P}\left(  \left\vert \mathbf{X}_{t}^{\mathbf{0}}\right\vert
<R\right)  \left(  1-\mathbb{P}\left(  \left\vert \mathbf{X}_{t}^{\mathbf{0}%
}\right\vert <R\right)  \right)  .
\]
Therefore, the key quantities in this example of $T_{0}$ and $\phi$ are:

\begin{enumerate}
\item the function $t\mapsto\mathbb{E}\left[  \left\vert \mathbf{X}%
_{t}^{\mathbf{0}}\right\vert ^{2}\right]  $

\item the function $t\mapsto\mathbb{P}\left(  \left\vert \mathbf{X}%
_{t}^{\mathbf{0}}\right\vert <R\right)  $.
\end{enumerate}

\subsection{Exact formulae in the control case}\label{app_sec_control}

In the control case $\mathbf{K}_{\eta}=\left\{  1,2\right\}  $, equation
(\ref{SDE}) read%
\[
d\mathbf{X}_{t}^{\mathbf{0}}=u\frac{\tau_{\eta}^{1-H}}{\sqrt{\Gamma\left(
2H+1\right)  }}d\left(  B_{t}^{H,1},B_{t}^{H,2}\right)
\]
namely
\[
\mathbf{X}_{t}^{\mathbf{0}}=u\frac{\tau_{\eta}^{1-H}}{\sqrt{\Gamma\left(
2H+1\right)  }}\left(  B_{t}^{H,1},B_{t}^{H,2}\right)  .
\]
We have%
\[
\mathbb{E}\left[  \left\vert \mathbf{X}_{t}^{\mathbf{0}}\right\vert
^{2}\right]  =\frac{2}{\Gamma\left(  2H+1\right)  }u^{2}\tau_{\eta}%
^{2-2H}t^{2H}=\frac{2}{\Gamma\left(  2H+1\right)  }u^{2}\tau_{\eta}^{2}\left(
\frac{t}{\tau_{\eta}}\right)  ^{2H}.
\]
In order to compute $\mathbb{P}\left(  \left\vert \mathbf{X}_{t}^{\mathbf{0}%
}\right\vert <R\right)  $, denote by $\mathbf{Z}=\left(  Z_{1},Z_{2}\right)  $
a standard normal vector and notice that $\left(  B_{t}^{H,1},B_{t}%
^{H,2}\right)  =t^{H}\mathbf{Z}$, hence $\mathbf{X}_{t}^{\mathbf{0}}%
=\frac{u\tau_{\eta}^{1-H}}{\sqrt{\Gamma\left(  2H+1\right)  }}t^{H}\mathbf{Z}%
$. Therefore%
\begin{align*}
\mathbb{P}\left(  \left\vert \mathbf{X}_{t}^{\mathbf{0}}\right\vert <R\right)
&  =\mathbb{P}\left(  \left\vert \mathbf{Z}\right\vert <\sqrt{\Gamma\left(
2H+1\right)  }u^{-1}\tau_{\eta}^{H-1}t^{-H}R\right) \\
&  =\mathbb{P}\left(  Z_{1}^{2}+Z_{2}^{2}<\Gamma\left(  2H+1\right)
u^{-2}\tau_{\eta}^{2H-2}t^{-2H}R^{2}\right) \\
&  =\mathbb{P}\left(  Y<\Gamma\left(  2H+1\right)  u^{-2}\tau_{\eta}%
^{2H-2}t^{-2H}R^{2}\right)
\end{align*}
where $Y$ has a Chi-squared distribution with two degrees of freedom, namely
with density
\[
f_{Y,2}\left(  y\right)  =\frac{e^{-y/2}}{2}1_{\left\{  y\geq0\right\}  }.
\]
For large $t$ we have a small value of $\epsilon=\Gamma\left(  2H+1\right)
u^{-2}\tau_{\eta}^{2H-2}t^{-2H}R^{2}$, hence%
\[
\mathbb{P}\left(  Y<\epsilon\right)  \sim f_{Y,2}\left(  0\right)
\epsilon=\frac{\epsilon}{2}%
\]
namely%
\begin{equation}
\mathbb{P}\left(  \left\vert \mathbf{X}_{t}^{\mathbf{0}}\right\vert <R\right)
\sim\frac{\Gamma\left(  2H+1\right)  }{2}u^{-2}\tau_{\eta}^{2H-2}t^{-2H}%
R^{2}=\frac{\Gamma\left(  2H+1\right)  }{2}\frac{R^{2}}{u^{2}\tau_{\eta}^{2}%
}\left(  \frac{t}{\tau_{\eta}}\right)  ^{-2H}%
\label{eq:prob_th_control}
\end{equation}
Notice that dimensional analysis is correct. It follows, for large $t$,
\begin{equation*}
m_{t}^{\phi}    =\mathbb{E}\left[  \left\langle T\left(  t\right)
,\phi\right\rangle \right]  \sim\frac{\Gamma\left(  2H+1\right)  }{2}%
\frac{R^{2}}{u^{2}\tau_{\eta}^{2}}\left(  \frac{t}{\tau_{\eta}}\right)
^{-2H}
\end{equation*}
and
\begin{equation}
\mathbb{E}\left[  \left(  \left\langle T\left(  t\right)  ,\phi\right\rangle
-m_{t}^{\phi}\right)  ^{2}\right]     \sim\frac{\Gamma\left(  2H+1\right)
}{2}\frac{R^{2}}{u^{2}\tau_{\eta}^{2}}\left(  \frac{t}{\tau_{\eta}}\right)
^{-2H}%
\label{eq:var_th_control}
\end{equation}
(since $1-\mathbb{P}\left(  \left\vert \mathbf{X}_{t}^{\mathbf{0}}\right\vert
<R\right)  \rightarrow0$).

\section{Diffusion constant of the approximate Brownian Motion}\label{sec:sigma_H_formula}

The main result of this paper is the fact that, in spite of the memory of the
processes involved, the behavior of a tracer is similar to a Brownian Motion
when the spatial structure of the fluid velocity field is complex enough. A
natural question is whether we can give a formula for the diffusion constant
of the approximate Brownian Motion. In this appendix we conjecture a formula
for this diffusion constant. The procedure to obtain it also clarifies the
intuition behind the fact itself of a Brownian behavior.

Recall that the tracer dynamics, starting from zero, is defined by equation
(\ref{SDE})%
\[
\mathbf{X}_{t}=uC\left(  \eta,\tau,H\right)  \sum_{\mathbf{k}\in
\mathbf{K}_{\eta}}\int_{0}^{t}\mathbf{\sigma}_{\mathbf{k}}\left(
\mathbf{X}_{s}\right)  \circ dB_{s}^{H,\mathbf{k}}.
\]
Over a very short time interval $\left[  t,t+t_{\eta}\right]  $, the
displacement can be approximated by%
\[
\mathbf{X}_{t+t_{\eta}}-\mathbf{X}_{t}\sim uC\left(  \eta,\tau,H\right)
\sum_{\mathbf{k}\in\mathbf{K}_{\eta}}\mathbf{\sigma}_{\mathbf{k}}\left(
\mathbf{X}_{t}\right)  \left(  B_{t+t_{\eta}}^{H,\mathbf{k}}-B_{t}%
^{H,\mathbf{k}}\right)
\]
(for $H=1/2$ we should consider the Stratonovich approximation but one can
show that it is not essential for the final result; we omit this point). We
perform further the following rough approximation (inspired by the
computations of Appendix \ref{appendix noise structure})%
\begin{align*}
\left\vert \mathbf{X}_{t+t_{\eta}}-\mathbf{X}_{t}\right\vert  & \sim uC\left(
\eta,\tau,H\right)  \sqrt{\frac{Card\left(  \mathbf{K}_{\eta}\right)  }{2}%
}t^{H}\\
& =u\frac{\sqrt{2}\tau^{1-H}}{\sqrt{\Gamma\left(  2H+1\right)  }}t^{H}.
\end{align*}

We want to discover that, approximately%
\[
\mathbf{X}_{t}\sim\sigma_{H}\mathbf{W}_{t}%
\]
where $\mathbf{W}_{t}$ is a 2d Brownian Motion. In particular, we want an
estimate of $\sigma_{H}$. The intuition is that, in an average time $t_{\eta}%
$, the tracer jumps from a Fourier component to the other (we mean that the
tracer is influenced mostly by a certain Fourier component, for a time of
order $t_{\eta}$, then mostly by another one). The increments%
\[
\mathbf{X}_{t_{\eta}},\mathbf{X}_{2t_{\eta}}-\mathbf{X}_{t_{\eta}}%
,\mathbf{X}_{3t_{\eta}}-\mathbf{X}_{2t_{\eta}},...
\]
will be approximately independent, since the various Fourier components are
affected by independent processes. Moreover, each increment has a length
$\left\vert \mathbf{X}_{\left(  i+1\right)  t_{\eta}}-\mathbf{X}_{it_{\eta}%
}\right\vert $ or order $\lambda\eta$, for a certain $\lambda>0$. Indeed,
these increments are the displacements when the tracer is affected by a
certain Fourier component, before jumping on another one, but the typical
"distance" to travel in order to jump from one to the other is of the order of
the wave-length of the sinusoidal components of the noise, possibly reduced by
a factor $\lambda$ (the intuition is that the tracer is on the "top" of a
cosine function, where the function takes approximately the value $\pm1$;
moving a little bit, just a portion of the wave-length $\eta$, it will be no
more on the top of that cosine component, but more near the top of another
component). 

Summarizing, at time steps $t_{\eta}$, we have a random walk with
displacements of size $\lambda\eta$. After $N$ time steps, the variance of the
position is of order $N\left(  \lambda\eta\right)  ^{2}$. In other words, at
time $Nt_{\eta}$ the square-average distance from the origin is $N\left(
\lambda\eta\right)  ^{2}$. Which should be also equal to $\sigma_{H}%
^{2}Nt_{\eta}$, hence
\[
\sigma_{H}\sim\frac{\lambda\eta}{\sqrt{t_{\eta}}}.
\]
But we have established above that
\[
\lambda\eta\sim u\frac{\sqrt{2}\tau^{1-H}}{\sqrt{\Gamma\left(  2H+1\right)  }%
}t_{\eta}^{H}.
\]
Hence%
\[
t_{\eta}\sim\left(  \frac{\lambda\eta\sqrt{\Gamma\left(  2H+1\right)  }}%
{\sqrt{2}\tau^{1-H}u}\right)  ^{1/H}%
\]
and finally%
\begin{equation}
\sigma_{H}\sim\left(  \frac{2}{\Gamma\left(  2H+1\right)  }\right)  ^{\frac
{1}{4H}}u^{\frac{1}{2H}}\tau_{\eta}^{\frac{1}{2H}-\frac{1}{2}}\lambda
^{1-\frac{1}{2H}}\eta^{1-\frac{1}{2H}}.
\label{eq:sigma_H}%
\end{equation}
We have not found an argument to predict the coefficient $\lambda$, but we can
show numerically that there exists a value providing a good fit between this
formula and numerical experiments.

\begin{remark}
Call $K=u\frac{\tau_{\eta}}{\eta}$ the Kubo number. Recall from Appendix \ref{sec:app_b}
that, for the Brownian Motion, $\sigma_{MB}\sim\sqrt{2}u_{\eta}\tau_{\eta
}^{\frac{1}{2}}$. Hence
\[
\sigma_{H}\sim\frac{1}{\sqrt{2}}\left(  \frac{2}{\Gamma\left(  2H+1\right)
}\right)  ^{\frac{1}{4H}}\frac{\sigma_{MB}}{K^{1-\frac{1}{2H}}}%
\]
Therefore $\sigma_{H}>>\sigma_{MB}$ if $K<<1.$
\end{remark}

\section*{Acknowledgements}
This research has been funded by the European Union, ERC NoisyFluid, No. 101053472.

\bibliographystyle{unsrt}
\bibliography{mybibfile}

\end{document}